\documentclass[12pt]{article} 

\usepackage[a4paper, top=2.5cm, bottom=2.5cm, left=2.5cm, right=2.5cm]{geometry}
\usepackage{authblk}  
\usepackage{titlesec} 
\titleformat{\section}{\bfseries\large}{\thesection}{1em}{}
\titleformat{\subsection}{\bfseries\normalsize}{\thesubsection}{1em}{}

\usepackage{amsmath}
\usepackage{amsthm}
\usepackage{bm}
\usepackage{amssymb}
\usepackage{graphicx}
\usepackage{subcaption}
\usepackage{lipsum}
\usepackage[normalem]{ulem}
\usepackage{cite}
\usepackage{hyperref}
\usepackage{enumerate}

\usepackage{tikz}
\usetikzlibrary{graphs,patterns,decorations.markings,arrows,matrix,shapes.misc}

\theoremstyle{plain}
\theoremstyle{definition}
\newtheorem{thm}{Theorem}
\newtheorem{lem}{Lemma}

\newtheorem{prop}{Proposition}
\newtheorem{remark}{Remark}

\DeclareMathOperator*{\Res}{Res}
\newcommand{\order}[1]{\mathcal{O}\left(#1\right)}
\newcommand{\mrm}[1]{\mathrm{#1}}
\newcommand{\ket}[1]{\vert #1 \rangle}
\newcommand{\bra}[1]{\langle #1 \vert}
\newcommand{\braket}[1]{\langle #1 \rangle}
\newcommand{\eq}[1]{\begin{equation}\begin{split}#1 \end{split}\end{equation}}
\newcommand{\eqnn}[1]{\begin{equation*}\begin{split}#1 \end{split}\end{equation*}}
\newcommand{\eqs}[1]{\begin{align} #1 \end{align}}
\newcommand{\eqsnn}[1]{\begin{align*} #1 \end{align*}}
\newcommand{\vcenteredinclude}[1]{
    \vcenter{\hbox{\includegraphics[scale=0.25]{#1}}}
}

\newcommand{\ack}{\section*{Acknowledgments}} 
\newcommand{\address}[1]{\affil{\small\itshape #1}}
\newcommand{\ead}[1]{\thanks{Email: #1}}

\begin{document}

\title{\bfseries Integral formula for the propagator of the one-dimensional Hubbard model}

\author[1]{Taiki Ishiyama\thanks{Corresponding author}\ead{ishiyama.t.ad@m.isct.ac.jp}}
\author[1]{Kazuya Fujimoto}
\author[1]{Tomohiro Sasamoto}

\address{Department of Physics, Institute of Science Tokyo, 2-12-1 Ookayama, Meguro-ku, Tokyo 152-8551, Japan}

\date{\today}

\maketitle

\begin{abstract}
\noindent
We present an exact integral formula for the multi-particle propagator
of the one-dimensional Fermi--Hubbard model on an infinite lattice.
The proof is based on the nested Bethe ansatz without relying on the string hypothesis.
Our formula enables an explicit integral representation of the time evolution
of arbitrary finite-particle wave functions
and thereby provides a foundation for the exact analysis of nonequilibrium dynamics
in the Hubbard model.
It can further be applied to related open quantum models.
\end{abstract}

\section{Introduction}
The one-dimensional (1D) Fermi--Hubbard model serves as a fundamental platform for studying strongly correlated electron systems 
in one dimension, where quantum fluctuations enhanced by low dimensionality give rise to 
intriguing phenomena such as Tomonaga--Luttinger liquid behavior and spin-charge separation~\cite{Essler_hubbard}.
A distinctive feature of this model is its integrability.
Leveraging this property, since the seminal work of Lieb and Wu~\cite{Lieb_hubbard}, its spectral properties 
and thermodynamics have been extensively investigated using the nested Bethe ansatz.

In recent years, interest in the 1D Fermi--Hubbard model
has shifted significantly from stationary properties to nonequilibrium dynamics,
motivated by remarkable experimental and theoretical advancements.
On the experimental side,
state-of-the-art ultracold atom experiments have enabled the faithful simulation of real-time evolution 
in isolated quantum many-body systems~\cite{Bloch_2008,Gross_2017}.
In particular, Hubbard-type models can be naturally realized in such platforms,
where the dimensionality, particle statistics, and interaction strength can be widely tuned.
On the theoretical side, 
tensor-network simulations (TNS), including 
the time-evolving block decimation method~\cite{Vidal_2004} 
and time-dependent density matrix renormalization group~\cite{White_2004},
have emerged as efficient numerical tools for simulating the time evolution of 1D quantum many-body systems~\cite{Schollwock_2011,Paeckel_2019}.
Furthermore, generalized hydrodynamics (GHD)~\cite{Bertini_2016,Castro_2016,Doyon_2020} has been formulated,
which describes hydrodynamic behavior at Euler scales for general integrable systems.

While these theoretical frameworks
have successfully elucidated nonequilibrium properties of the 1D Fermi--Hubbard model~\cite{Kajala_2011,Langer_2012,Vidmar_2013,Karrasch_2016, Prosen_2012,Moca_2023,llievski_2017,llievski_2018,Fava_2020, Nozawa_2020,Nozawa_2021},
we note that the accessible time scale of TNS is strongly limited by entanglement growth,
and GHD primarily describes the average ballistic transport of conserved charges.
Consequently, exact microscopic approaches are indispensable to complement these frameworks,
particularly to capture phenomena beyond the regimes of validity of these methods.

Despite the integrability of the 1D Fermi--Hubbard model, performing exact microscopic calculations of
nonequilibrium dynamics remains a formidable challenge.
For simpler quantum integrable systems, such as the Lieb--Liniger Bose gas and the XXZ spin chain,
exact integral representations of the multi-particle propagator
on an infinite interval, often referred to as the Yudson representation~\cite{Yudson_dynamics1,Yudson_dynamics2}, 
are well established via the Bethe ansatz~\cite{Tracy_2008,Iyer_quench,Iyer_exact,Liu_quench,Guan_quench}.
These formulas enable an exact analysis of time evolution for finite particles on the infinite interval without relying on the string hypothesis.
In addition to the infinite interval case, the exact propagator in finite volume has been obtained for the XXZ spin chain~\cite{Feher_2019}.
Recently, several studies have attempted to investigate the nonequilibrium dynamics of finite-density initial states
in the XXZ spin chain by taking the thermodynamic limit of these integral representations
\cite{Feher_2019,Saenz2022, Fujimoto_quantum}.
In parallel, closely related integral formulas for multi-particle propagators have been developed independently
in the field of integrable stochastic interacting systems~\cite{Schutz_1997,Tracy_integral, Prolhac_2011,Tracy_2013}.
In this context, various studies have demonstrated that the use of such formulas together with the Markov duality~\cite{Giardina}
enables the analysis of nonequilibrium dynamics at finite density~\cite{Tracy_Fredholm,Tracy_asymptotics,Borodin_duality,Derrida_current,Borodin_2016}.
However, no exact integral representation of the propagator is known for the 1D Fermi--Hubbard model.

In this work, we derive an exact integral formula for the multi-particle propagator of the 1D Fermi--Hubbard model on an infinite lattice.
The formula is expressed as a multiple contour integral of the Bethe wave function with appropriately chosen integration contours.
Unlike the Lieb--Liniger model and the XXZ spin chain,
the presence of internal degrees of freedom in the 1D Fermi--Hubbard model necessitates the use of the nested Bethe ansatz~\cite{Lieb_hubbard,Essler_hubbard}.
Consequently, the derivation of the formula is technically more involved, requiring us to explicitly employ the nested structure of the Bethe wave function.
The formula provides an explicit representation of the time evolution 
of arbitrary finite-particle wave functions on the infinite lattice. 
It can furthermore be applied to open quantum systems whose effective description is given by the Hubbard model with complex 
interaction~\cite{Medvedyeva_exact,Nakagawa_2021,Yoshida_liouvillian, Marche_universality}.
Thus our formula establishes a foundation for the exact microscopic analysis of nonequilibrium dynamics in these systems.

The paper is organized as follows. In section~\ref{sec:setup}, we define the model and the multi-particle propagator.
In section~\ref{sec:bethe}, we briefly review the nested Bethe ansatz for the 1D Fermi--Hubbard model~\cite{Lieb_hubbard,Essler_hubbard}.
In section~\ref{sec:result}, we derive an integral formula for the propagator
using the nested Bethe ansatz~\cite{Lieb_hubbard, Essler_hubbard}.
We conclude the paper in section~\ref{sec:conclusion}.

\section{Setup: the 1D Fermi--Hubbard model}\label{sec:setup}
In this work, we consider the 1D Fermi--Hubbard model on an infinite lattice $\mathbb{Z}$.
For brevity, we refer to it simply as the Hubbard model throughout the rest of this paper.
The Hamiltonian is given by
\eq{
\hat{H} := - \sum_{x\in \mathbb{Z}}\sum_{\sigma= \downarrow,\uparrow} \left(\hat{c}^\dagger_{x,\sigma} \hat{c}_{x+1,\sigma}
+ \hat{c}^\dagger_{x+1,\sigma}\hat{c}_{x,\sigma}\right) + 4u
\sum_{x\in \mathbb{Z}} \hat{n}_{x,\downarrow} \hat{n}_{x,\uparrow}.
\label{eq:Hubbard}
}
Here, $\hat{c}_{x,\sigma}$ ($\hat{c}_{x,\sigma}^\dagger$) and $\hat{n}_{x,\sigma}:=\hat{c}^\dagger_{x,\sigma} \hat{c}_{x,\sigma}$
denote the annihilation (creation) operator and the number operator of fermions at site $x$ with spin $\sigma$, respectively.
The on-site interaction strength is parameterized by $4u$, where the factor $4$ is introduced for later convenience.
While the interaction parameter $u$ is usually real in \eqref{eq:Hubbard}, 
we allow it to take complex values $u\in \mathbb{C}\setminus\{0\}$ throughout this work. 
We note that the non-interacting case is excluded due to technical reasons 
that will become clear in section~\ref{sec:result}.

The Hamiltonian conserves the total particle number $\hat{N}$ and the number of down-spin particles $\hat{N}_{\downarrow}$,
\eq{
[\hat{H}, \hat{N}] = [\hat{H}, \hat{N}_{\downarrow}] = [\hat{N},\hat{N}_{\downarrow}]= 0,
}
where $\hat{N} := \sum_{x\in \mathbb{Z}} \sum_{\sigma = \downarrow,\uparrow}\hat{n}_{x,\sigma}$ and
$\hat{N}_{\downarrow}:=\sum_{x\in \mathbb{Z}}\hat{n}_{x,\downarrow}$.
Hereafter, we focus on the sector with $N$ particles and $M$ down-spin particles,
assuming $N, M > 0$.

The quantity of interest is the propagator
\eq{
\psi_t(\bm{x};\bm{a}\vert \bm{y};\bm{b}) :=\langle \bm{x};\bm{a}\vert e^{-it\hat{H}} \vert \bm{y};\bm{b}\rangle,
\label{eq:def_pro}
}
where we define the Wannier state with the vacuum state $\ket{0}$ as
\eq{
\ket{\bm{x};\bm{a}} := \hat{c}^\dagger_{x_{N},a_N}\cdots \hat{c}^\dagger_{x_1,a_1}\ket{0}.
}
The propagator describes the transition amplitude that the system, initially prepared in the state $\vert \bm{y};\bm{b}\rangle$, evolves into
the state $\ket{\bm{x};\bm{a}}$ at time $t$.
In terms of the propagator, any $N$-particle state $\ket{\Phi(t)}$ can be expanded as follows,
\eqs{
\ket{\Phi(t)} 
&= 
\frac{1}{(N!)^2 }\sum_{\bm{x},\bm{y}\in \mathbb{Z}^N }\sum_{\bm{a},\bm{b}\in \{\downarrow,\uparrow\}^N }
 \ket{\bm{x};\bm{a} }\bra{\bm{x};\bm{a}}  e^{-it\hat{H}}
\ket{\bm{y};\bm{b}}
\bra{\bm{y};\bm{b}} \Phi(0)\rangle
\\
&= \frac{1}{(N!)^2} \sum_{\bm{x},\bm{y}\in \mathbb{Z}^N }\sum_{\bm{a},\bm{b}\in \{\downarrow,\uparrow\}^N }
\psi_t(\bm{x};\bm{a}\vert \bm{y};\bm{b}) \langle \bm{y};\bm{b}\vert \Phi(0)\rangle \ket{\bm{x};\bm{a}}.
}
Namely, one can exactly calculate any $N$-particle state $\ket{\Phi(t)}$ 
if one obtains the exact solution of $\psi_t(\bm{x};\bm{a}\vert \bm{y};\bm{b})$.

The propagator satisfies the following initial condition,
\eq{
\psi_t(\bm{x};\bm{a}\vert \bm{y};\bm{b})\vert_{t=0}  = \det[\delta_{x_j, y_k}\delta_{a_j, b_k}]_{j,k=1}^N,
\label{eq:initial}
}
and 
the Schr\"{o}dinger equation of the one-dimensional Hubbard model in first quantization~\cite{Essler_hubbard},
\eqs{
&i \frac{d}{dt} \psi_t(\bm{x};\bm{a}|\bm{y};\bm{b})=H_{N}\psi_t(\bm{x};\bm{a}|\bm{y};\bm{b}), \label{eq:hubbard}
\\
&H_{N}:= -\sum_{j=1}^{N}(\Delta_{j}^{+} + \Delta_j^-) +4u \sum_{1\leq j<k\leq N} \delta_{x_j,x_k}, 
}
where we define the shift operator
\eq{
\Delta^{\pm}_j \psi_t(\bm{x};\bm{a}\vert \bm{y};\bm{b}):= \psi_t(\bm{x}\pm \bm{e}_j;\bm{a}\vert \bm{y};\bm{b}).
}
These follow from the definition of the propagator~\eqref{eq:def_pro}.

In section~\ref{sec:result}, we solve this initial value problem and derive a multiple contour integral formula for 
$\psi_t(\bm{x};\bm{a}|\bm{y};\bm{b})$ in the sector with $N$ particles and $M$ down-spins. 
Our derivation employs the nested Bethe ansatz~\cite{Lieb_hubbard,Essler_hubbard}, which is reviewed in section~\ref{sec:bethe}.

We remark that the interaction strength is assumed to take a complex value
$u \in \mathbb{C}\setminus\{0\}$ throughout this work.
This assumption is motivated by the fact that Hubbard models with complex interaction strengths
naturally arise in several open quantum systems.
For instance, the tight-binding chain with dephasing noise is related to the Hubbard model with 
imaginary interaction~\cite{Medvedyeva_exact}.
As another example, the Hubbard model (with real interaction)
subject to two-body loss is effectively described by the Hubbard model with complex interaction~\cite{Nakagawa_2021,Yoshida_liouvillian, Marche_universality}.
See appendix~\ref{app:open} for the application of the propagator \eqref{eq:def_pro} to the analysis of these open systems.

\section{Nested Bethe ansatz}\label{sec:bethe}
It is known that the stationary Schr\"odinger equation of the Hubbard model,
\eq{
H_N \phi(\bm{x};\bm{a}) = E\,\phi(\bm{x};\bm{a}),
\label{eq:stationary}
}
admits exact solutions via the nested Bethe ansatz \cite{Lieb_hubbard,Essler_hubbard}.
In this section, we briefly review the nested Bethe ansatz in the Hubbard model.
The notation used here basically follows chapter 3 of \cite{Essler_hubbard}.

In the sector with $N$ particles and $M$ down spins, the Bethe ansatz solution is
parameterized by the charge rapidities $\bm{z}=(z_1,\dots,z_N)$ and
the spin rapidities $\bm{\lambda}=(\lambda_1,\dots,\lambda_M)$.
Its explicit form reads
\eq{
E(\bm{z}) := -\sum_{j=1}^N (z_j+1/z_j)\label{eq:energy}
}
and
\eq{
&\phi(\bm{x};\bm{a} \vert \bm{z};\bm{\lambda}) := \sum_{P\in S_N} \mathrm{sgn}(P)\langle \bm{a} | \bm{s}P;\bm{\lambda}\rangle
 \prod_{j=1}^N z_{P(j)}^{x_j}
 \label{eq:wave}.
}
Here the scattering amplitude for $N$ particles is defined as
\eq{ 
\langle \bm{a} | \bm{s}P;\bm{\lambda}\rangle &:= 
    \sum_{R\in S_M} \prod_{\substack{1\leq m<n\leq M \\ R^{-1}(n) < R^{-1}(m)}} \Big[\frac{\lambda_m-\lambda_n+2iu}{\lambda_m-\lambda_n-2iu}\Big]
    \\
   & \times
    \prod_{l=1}^M \Big[
    \frac{2iu}{\lambda_{R(l)} - s_{P(\alpha_l)} +iu } \prod_{j = 1}^{\alpha_l-1} \frac{\lambda_{R(l)} - s_{P(j)} -iu}{\lambda_{R(l)} - s_{P(j)} +iu}\Big],
    \label{eq:spin_wave}
}
where $\bm{s}P := (s_{P(1)},\cdots,s_{P(N)})$ and the variable $s_j$ is defined as
\eq{
s_j := (z_j - 1/z_j)/(2i).
\label{eq:s_j}
}
In these expressions, $S_N$ and $S_M$ denote the permutation groups of $N$ and $M$ elements, respectively.
In \eqref{eq:wave}, we assume $x_1\leq \dots \leq x_N$ without loss of generality, utilizing the antisymmetry of the fermionic wave function.
In \eqref{eq:spin_wave}, 
$\alpha_l$ ($l=1,\dots,M$) denotes the position of the $l$-th down spin in the sequence $\bm{a}=(a_1,\cdots,a_N)$.
Since the amplitude \eqref{eq:spin_wave} depends on the spin configuration $\bm{a}$ and corresponds to a vector component
in the auxiliary spin space spanned by $\{\ket{\bm{a}}\}$, we adopt the braket notation.
We remark that the normalization of \eqref{eq:spin_wave} differs from that of equation (3.92) in \cite{Essler_hubbard}.

Equation~\eqref{eq:wave} has a nested structure.
It has the form of a Bethe wave function with respect to the particle coordinates,
while the corresponding scattering amplitude for $N$ particles, given in \eqref{eq:spin_wave},
again takes the form of a Bethe wave function, now defined in an auxiliary spin space.
Specifically, \eqref{eq:spin_wave} can be identified with the Bethe ansatz wave function of 
the inhomogeneous XXX spin chain~\cite{Essler_hubbard,Korepin_1993}; see appendix~\ref{app:yang} for this relation.
In this interpretation, the factor
\eq{
\prod_{\substack{1\leq m<n\leq M \\ R^{-1}(n) < R^{-1}(m)}} \Big[\frac{\lambda_m-\lambda_n+2iu}{\lambda_m-\lambda_n-2iu}\Big]
}
describes the scattering amplitude for $M$ magnons, while
\eq{
\prod_{l=1}^M \Big[
    \frac{2iu}{\lambda_{R(l)} - s_{P(\alpha_l)} +iu } \prod_{j = 1}^{\alpha_l-1} \frac{\lambda_{R(l)} - s_{P(j)} -iu}{\lambda_{R(l)} - s_{P(j)} +iu}\Big]
}
corresponds to the ``plane waves'' for $M$ magnons in the XXX spin chain subject to the inhomogeneities $s_j$.
In particular, if we take all $s_j = 0$, \eqref{eq:spin_wave} reduces to the Bethe wave function of 
the homogeneous XXX spin chain.

The scattering amplitudes \eqref{eq:spin_wave} for a permutation $P$ and its adjacent transposition $P\Pi_{n,n+1}$ 
are related through the $Y$-operator
\eq{
Y_{n,n+1}(\mu) := \frac{\mu \Pi_{n,n+1} + 2iu}{\mu + 2iu} \label{eq:yang}
}
as follows,
\eq{
\ket{\bm{s}P\Pi_{n,n+1};\bm{\lambda}} = Y_{n,n+1}(s_{P(n)}- s_{P(n+1)})\ket{\bm{s}P;\bm{\lambda}}.
\label{eq:sca}
}
Here, the action of a general permutation operator $Q$ on the basis vectors is defined by
\eq{
Q \ket{\bm{a}} := \ket{a_{Q^{-1}(1)}, \cdots, a_{Q^{-1}(N)}}.
}
The relation \eqref{eq:sca} guarantees that the Bethe ansatz solution given in \eqref{eq:energy} and \eqref{eq:wave} 
satisfies the stationary Schr\"{o}dinger equation \eqref{eq:stationary} on the infinite lattice; 
see appendix~B in chapter~3 of \cite{Essler_hubbard}. 
The fact that the scattering amplitudes \eqref{eq:spin_wave} indeed satisfy \eqref{eq:sca} follows from 
the algebraic structure of the inhomogeneous XXX spin chain~\cite{Essler_hubbard,Korepin_1993}; see appendix \ref{app:yang}.

As shown in the next section, the relation~\eqref{eq:sca} plays an important role 
in the proof of the exact expression for the propagator. 
It allows one to express the scattering amplitude for an arbitrary permutation $P$ 
in terms of that for the identity permutation $\mrm{id}$ and a sequence of $Y$-operators. 
Consequently, problems in which scattering amplitudes for general permutations appear 
reduces to the identity case, 
thereby simplifying the analysis.

\begin{remark}
On the infinite lattice, the Bethe ansatz solution is not quantized.
In other words,
$E(\bm{z})$ and $\phi(\bm{x};\bm{a}\vert \bm{z};\bm{\lambda})$ solve the stationary Schr\"{o}dinger equation
for arbitrary rapidities $\bm{z}$ and $\bm{\lambda}$ except in cases 
where either $E(\bm{z})$ or $\phi(\bm{x};\bm{a}\vert \bm{z};\bm{\lambda})$ becomes singular.
This lack of quantization is inconvenient for studying thermodynamic properties of equilibrium states.
To obtain well-defined thermodynamic behavior, one usually considers a finite lattice with periodic boundary conditions, 
which leads to the Lieb-Wu equations that quantize the rapidities~\cite{Lieb_hubbard,Essler_hubbard}.
However, as we will see in the next section, it is advantageous to work on the infinite lattice for the purpose of deriving the exact expression for the propagator.
\end{remark}

\section{Integral formula for the propagator}\label{sec:result}
Here we present a multiple contour integral formula for the propagator~\eqref{eq:def_pro}.

In \eqref{eq:def_pro}, we assume that $x_1\leq \cdots \leq x_N$ ($y_1\leq \cdots \leq y_N$)
and $(x_j,a_j)\neq (x_k,a_k)$ ($(y_j,b_j)\neq(y_k,b_k)$) for any $j\neq k$ without loss of generality
since the propagator is antisymmetric under the exchange of the pairs $(x_j,a_j)$ (or $(y_j,b_j)$).
Then we can state the following theorem.

\begin{thm}\label{thm}
The propagator~\eqref{eq:def_pro} can be expressed as
\eq{
\psi_t(\bm{x};\bm{a}\vert \bm{y};\bm{b}) &= \Big[\prod_{j=1}^N \oint_{|z_j|=r^{N-j}} \frac{dz_j}{2\pi i z_j} z_j^{-y_j}\Big] 
\Big[\prod_{k=1}^M \oint_{\Gamma_{\bm{s}}} \frac{d\lambda_k }{2\pi i (\lambda_k - s_{\beta_k} -iu)} 
\prod_{l=1}^{\beta_k-1} \frac{\lambda_k-s_l + iu}{\lambda_k -s_l-iu } \Big]
\\
&\quad \times
e^{-iE(\bm{z}) t}\phi(\bm{x};\bm{a}\vert \bm{z};\bm{\lambda} ),
\label{eq:formula}
}
where $E(\bm{z})$ and $\phi(\bm{x};\bm{a}\vert \bm{z};\bm{\lambda})$
are given in \eqref{eq:energy} and \eqref{eq:wave}, respectively.
$\beta_k$ denotes the position of the $k$-th down spin
in $(b_1,\cdots,b_N)$.
All integration contours in \eqref{eq:formula} are oriented counterclockwise.
We set $r\ll 1$ so that the poles of $1/(s_j-s_k+2iu)$ do not lie inside the $z_j$-contour for $j<k$; the function $s_j$ is defined
in \eqref{eq:s_j}.
The contour $\Gamma_{\bm{s}}$ consists of sufficiently small closed curves 
enclosing the set of points $\{s_j+iu\}_{j=1}^N$ such that the only possible poles inside the contour are these points
(see figure~\ref{fig:contour}).
The subscript $\bm{s} = (s_1, \cdots, s_N)$ is introduced to highlight the explicit dependence of the contour $\Gamma_{\bm{s}}$
on these variables.
\end{thm}

\begin{figure*}[htbp] 
\centering
\begin{tikzpicture}[>=stealth, scale=1.0]
    \tikzset{
        cross/.style={
            path picture={ 
                \draw[black, thick]
                (path picture bounding box.south east) -- (path picture bounding box.north west)
                (path picture bounding box.south west) -- (path picture bounding box.north east);
            },
            minimum size=4pt,
            inner sep=0pt,
            outer sep=0pt
        }
    }
    \tikzset{
        contour/.style={
            thick, 
            blue!80!black,
            postaction={decorate},
            decoration={
                markings,
                mark=at position 0.25 with {\arrow{>}}
            }
        }
    }
    \def\u{0.8} 

    \begin{scope}[xshift=-4cm]
        \draw[->, gray] (-3.5, 0) -- (3.5, 0) node[right, black] {$\mathrm{Re}\,\lambda$};
        \draw[->, gray] (0, -2.0) -- (0, 2.5) node[above, black] {$\mathrm{Im}\,\lambda$};
        
        \node[black, font=\bfseries] at (-3.5, 2.5) {(a)};

        \def\points{
            -2.5/1.0/s_1, 
            1.8/0.3/s_2, 
            -0.5/0.6/s_3
        }

        \foreach \x/\y/\txt in \points {
            \node[cross] at (\x, \y+\u) {};
            
            \draw[contour] (\x, \y+\u) circle (0.2);

            \node[cross, red!70!black] at (\x, \y-\u) {};
        }
        
        \node[below, font=\small] at (-2.5, 1.0+\u-0.15) {$s_1+iu$};
        \node[below, font=\small] at (-2.5, 1.0-\u-0.15) {$s_1-iu$};
    \end{scope}

    \begin{scope}[xshift=4cm]
        \draw[->, gray] (-2.5, 0) -- (2.5, 0) node[right, black] {$\mathrm{Re}\,\lambda$};
        \draw[->, gray] (0, -2.0) -- (0, 2.5) node[above, black] {$\mathrm{Im}\,\lambda$};

        \node[black, font=\bfseries] at (-2.5, 2.5) {(b)};

        \node[cross] at (0, \u) {};
        \node[right, font=\small] at (0.1, \u) {$iu$ };
        
        \draw[contour] (0, \u) circle (0.2);

        \node[cross, red!70!black] at (0, -\u) {};
        \node[right, font=\small] at (0.1, -\u) {$-iu$};
    \end{scope}
\end{tikzpicture}
\caption{The integration contour $\Gamma_{\bm{s}}$ in the complex $\lambda$-plane, indicated by blue circles with arrows.
(a) Inhomogeneous case, where $s_j$ are distinct. (b) Homogeneous case ($s_j=0$).
The situation in Theorem~\ref{thm} corresponds to the inhomogeneous case (a), 
since $s_j$ is defined in \eqref{eq:s_j} and the choice of the $z$-contour, $|z_j|=r^{N-j}$ with $r\ll 1$, implies that $s_j$ are well-separated.
In subsection~\ref{sub:spin}, however, we treat $\bm{s} = (s_1,\cdots,s_N)$ as general complex numbers subject 
only to the condition $s_j -s_k\neq 2iu$ for any $j,k$, which includes the homogeneous case (b).}
\label{fig:contour}
\end{figure*}
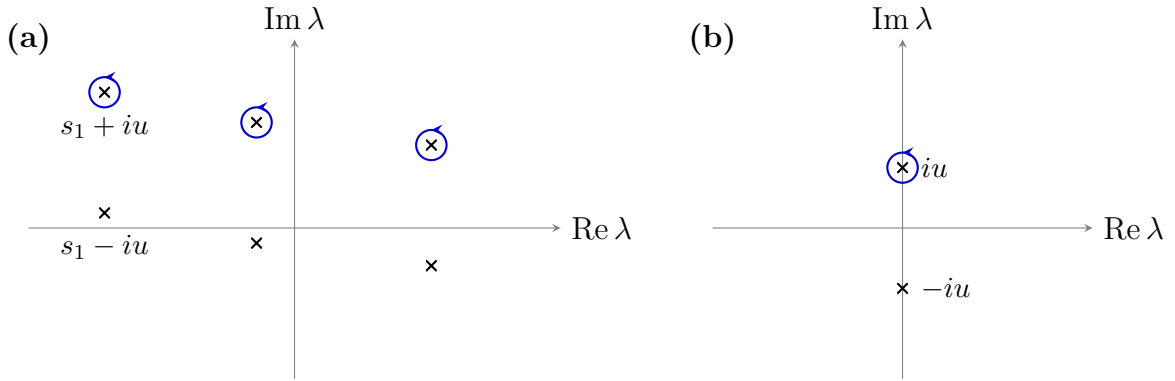

\begin{remark}
Theorem~\ref{thm} holds for any $u \in \mathbb{C} \setminus \{0\}$. 
The case $u=0$ is excluded because the function $\phi(\bm{x};\bm{a}\vert \bm{z};\bm{\lambda})$ vanishes
and the contour $\Gamma_{\bm{s}}$ is no longer viable as $u\to0$ 
due to the coalescence of the points $s_j + iu$ and $s_j - iu$.
\end{remark}

\begin{remark}
In our previous work~\cite{Ishiyama}, 
we derived the integral formula for the propagator for the specific case of $N=2$ and $M=1$ with $y_1 =y_2$.
Theorem~\ref{thm} constitutes a natural generalization to arbitrary $N$ and $M$, 
as well as general initial configurations $\bm{y}$ and $\bm{b}$.
\end{remark}

Theorem~\ref{thm} provides an exact multiple contour integral representation 
of the multi-particle propagator on the infinite lattice. 
Equation~\eqref{eq:formula} renders the time evolution of arbitrary 
finite-particle wave functions analytically accessible in the Hubbard model. 
Furthermore, Theorem~\ref{thm} is applicable to related open quantum models, 
as detailed in appendix~\ref{app:open}. 
Our formula thus provides the foundation for the exact analysis of 
nonequilibrium dynamics in these systems. 

From a practical point of view, our formula is given as an $(N+M)$-fold contour integral, 
and is therefore particularly well suited for few-particle problems in the Hubbard model. 
Nevertheless, in the tight-binding chain with dephasing noise, 
which can be mapped to the Hubbard model~\cite{Medvedyeva_exact} (see appendix~\ref{app:open}), 
few-point correlation functions for finite-density initial conditions
can be computed using the few-particle integral formula. 
This is due to the property known as duality~\cite{Medvedyeva_exact}, 
which reduces the calculation of $n$-point correlation functions in the tight-binding chain 
with dephasing noise to that of $n$-particle density matrices. 
In particular, in our previous work, the case $N=2$ and $M=1$ 
was used to derive exact results for the density profiles 
under domain-wall and alternating initial conditions. 
We refer the interested reader to Ref.~\cite{Ishiyama} for the details.

\bigskip
\noindent
\textbf{Proof of Theorem~\ref{thm}.}~
Our proof consists of the following two steps;
\begin{enumerate}[(a)]
\item the right-hand side of \eqref{eq:formula} satisfies the Schr\"{o}dinger equation~\eqref{eq:hubbard},
\item the right-hand side of \eqref{eq:formula} satisfies the initial condition~\eqref{eq:initial}.
\end{enumerate}
\textbf{Step (a)}. This can be proved by noting the fact that $e^{-iE(\bm{z})t}\phi(\bm{x};\bm{a}|\bm{z};\bm{\lambda})$ is 
the solution to the Schr\"{o}dinger equation~\eqref{eq:hubbard} from the nested Bethe ansatz \cite{Lieb_hubbard,Essler_hubbard}.
\medskip

\noindent
\textbf{Step (b)}. 
The proof of Step (b) is technically involved, and we begin by outlining its strategy. 
We first rewrite RHS of \eqref{eq:formula} at $t=0$ by exploiting the nested structure of the Bethe wave function~\eqref{eq:wave}: 
\eqs{
&\text{RHS of \eqref{eq:formula}}|_{t=0} =  \sum_{P\in S_N} \mathrm{sgn}(P)
\prod_{j=1}^N \Big[\oint_{|z_j|=r^{N-j}} \frac{dz_j}{2\pi i z_j} z_j^{-y_j}\Big] A^{(\bm{s})}_P(\bm{a}\vert \bm{b}) \prod_{j=1}^N z_{P(j)}^{x_j},
\label{eq:charge}
\\
&A^{(\bm{s})}_P(\bm{a}\vert \bm{b}) := \prod_{k=1}^M \Big[\oint_{\Gamma_{\bm{s}}} \frac{d\lambda_k }{2\pi i (\lambda_k - s_{\beta_k} -iu)} 
\prod_{l=1}^{\beta_k-1} \frac{\lambda_k-s_l + iu}{\lambda_k -s_l-iu } \Big]
\braket{\bm{a}\vert \bm{s}P;\bm{\lambda}},
\label{eq:A_P}
}
where $\braket{\bm{a}\vert \bm{s}P;\bm{\lambda}}$ is defined in \eqref{eq:spin_wave}.
Note that $A^{(\bm{s})}_P(\bm{a}\vert \bm{b})$ depends on $z_j$ solely through $s_j=(z_j-1/z_j)/(2i)$.

The goal of Step (b) is to show that \eqref{eq:charge} is equivalent to \eqref{eq:initial}
by first performing the $\lambda$-integrals in \eqref{eq:A_P} and subsequently the $z$-integrals in \eqref{eq:charge}. 
In subsection~\ref{sub:spin}, we carry out the $\lambda$-integrals and derive a representation of 
$A^{(\bm{s})}_P(\bm{a}\vert \bm{b})$ that does not involve the $\lambda$-integrals, 
as stated in Lemma~\ref{lem1} below. 
In subsection~\ref{sub:charge}, substituting this representation into \eqref{eq:charge}, 
we evaluate the $z$-integrals, thereby establishing the equivalence between \eqref{eq:charge} and \eqref{eq:initial}.
This result is formulated as Lemma~\ref{lem2} below.
In what follows, we present Lemma~\ref{lem1} and Lemma~\ref{lem2} and outline their proofs.

In subsection~\ref{sub:spin}, we treat $\bm{s}$ as independent complex parameters 
subject to the condition $s_j - s_k \neq 2iu$ for any $j\neq k$, and prove the following lemma.
\begin{lem}\label{lem1}
Let a permutation $P$ be decomposed into a product of adjacent transpositions as
$P=\Pi_{j_n,j_{n}+1}\cdots \Pi_{j_1,j_1+1}$.
Then the amplitude $A^{(\bm{s})}_P(\bm{a}\vert \bm{b})$ can be written as
\eq{
A^{(\bm{s})}_{P}(\bm{a}\vert \bm{b})  =  
\bra{ \bm{a}} Y_{j_1,j_1+1}(s_{P_1(j_1)} - s_{P_1(j_1+1)}) \cdots Y_{j_n, j_{n}+1} (s_{P_n(j_n)} - s_{P_n(j_n+1)})\ket{\bm{b}},
\label{eq:A_P_1}
}
where we define $P_k := \Pi_{j_n,j_n+1}\cdots \Pi_{j_{k+1},j_{k+1}+1}$ for $k=1,\dots,n-1$ and $P_n := \mrm{id}$.
The $Y$-operator is defined in \eqref{eq:yang}.
\end{lem}
\noindent
Note that the expression in \eqref{eq:A_P_1} no longer contains the $\lambda$-integrals, 
allowing for a straightforward evaluation of the $z$-integrals in \eqref{eq:charge}. 
Moreover, as explained in subsection~\ref{sub:charge}, 
this expression admits a graphical representation, providing an intuitive understanding 
of the pole structure of $A^{(\bm{s})}_P(\bm{a}\vert \bm{b})$. 
The proof of Lemma~\ref{lem1} proceeds by explicitly evaluating the $\lambda$-integrals 
for $A^{(\bm{s})}_{\mrm{id}}(\bm{a}\vert \bm{b})$ 
and establishing $\eqref{eq:A_P_1}$ for the identity permutation
(see Proposition~\ref{prop} below), together with the relation~\eqref{eq:sca} 
for the scattering amplitude~\eqref{eq:spin_wave}.

In subsection~\ref{sub:charge}, using Lemma~\ref{lem1}, we prove the following lemma, thereby completing the proof of Step (b).
\begin{lem}\label{lem2}
The following identity holds:
\eq{
\sum_{P\in S_N}\mrm{sgn}(P) \prod_{j=1}^N \Big[ \oint_{|z_j|=r^{N-j} } \frac{dz_j}{2\pi i z_j}z_j^{-y_j} \Big]
A^{(\bm{s})}_P(\bm{a}\vert \bm{b}) \prod_{j=1}^N  z_{P(j)}^{x_j} = \det[\delta_{x_j,y_k}\delta_{a_j,b_k}]_{j,k=1}^N,
\label{eq:identity_c_1}
}
where $A^{(\bm{s})}_P(\bm{a}\vert \bm{b})$ is given in \eqref{eq:A_P_1}, and
we set $r\ll1 $ so that the poles of $1/(s_j - s_k +2iu)$ do not lie inside the $z_j$-contour if $j<k$.
The function $s_j$ is defined in \eqref{eq:s_j}.
\end{lem}
\noindent
The proof of Lemma~\ref{lem2} follows from a similar argument to that of Proposition~\ref{prop}, 
combined with the representation~\eqref{eq:A_P_1} obtained in Lemma~\ref{lem1}.

\begin{remark}
Note that the representation of $P$ as a product of adjacent transpositions is not unique.
Hence one must verify that the value of $A^{(\bm{s})}_P(\bm{a}\vert \bm{b})$ in \eqref{eq:A_P_1} does not depend on the specific decomposition.
However this consistency problem is guaranteed by the (braid) Yang-Baxter relation,
\eq{
Y_{j,k}(\mu-\nu) Y_{k,l}(\lambda-\nu) Y_{j,k}(\lambda-\mu) = Y_{k,l}(\lambda-\mu)Y_{j,k}(\lambda-\nu)Y_{k,l}(\mu-\nu)
\label{eq:yang-baxter}
}
together with the unitarity relation,
\eq{
Y_{j,k}^{-1}(\lambda) = Y_{j,k}(-\lambda).
}
We refer the reader to appendix~C in chapter~3 of \cite{Essler_hubbard} for further properties of the $Y$-operator.
\end{remark}

\subsection{Proof of Lemma~\ref{lem1}}\label{sub:spin}
We prove Lemma~\ref{lem1} by exploiting the relation \eqref{eq:sca} and by explicitly evaluating the amplitude 
$A^{(\bm{s})}_{P}(\bm{a}\vert \bm{b})$ for the identity permutation $P=\mrm{id}$.
As mentioned earlier, throughout this subsection, we treat $\bm{s}$ as independent complex parameters satisfying 
$s_j -s_k\neq 2iu$ for any $j\neq k$ although $s_j$ is originally defined by $s_j = (z_j-1/z_j)/(2i)$.

For a permutation $P = \Pi_{j_n,j_{n}+1}\cdots \Pi_{j_1,j_1+1}$,
the scattering amplitude $\langle \bm{a}\ket{\bm{s}P;\bm{\lambda}}$ in
\eqref{eq:spin_wave} can be rewritten as
\eqnn{
\langle \bm{a}\ket{\bm{s}P;\bm{\lambda}} = \sum_{\bm{a}'\in \{\uparrow,\downarrow\}^N }
\bra{\bm{a}} Y_{j_1,j_1+1}(s_{P_1(j_1)} - s_{P_1(j_1+1)}) \cdots Y_{j_n,j_{n}+1} (s_{P_n(j_n)} - s_{P_n(j_n+1)})\ket{\bm{a}'}
\langle \bm{a}'\ket{\bm{s};\bm{\lambda}},
}
where we used \eqref{eq:sca} and the notation introduced in
Lemma~\ref{lem1}.
Substituting this expression into \eqref{eq:A_P}, the amplitude
$A^{(\bm{s})}_P(\bm{a}\vert \bm{b})$ can be written as
\eqnn{
A^{(\bm{s})}_P(\bm{a}\vert \bm{b}) =
\sum_{\bm{a}'\in \{\uparrow,\downarrow\}^N}
A^{(\bm{s})}_{\mrm{id}}(\bm{a}'\vert \bm{b})
\bra{\bm{a}} Y_{j_1,j_1+1}(s_{P_1(j_1)} - s_{P_1(j_1+1)}) \cdots Y_{j_n,j_{n}+1} (s_{P_n(j_n)} - s_{P_n(j_n+1)})\ket{\bm{a}'}.
}
Thus it suffices to show the following proposition.

\begin{prop}\label{prop}
The following identity holds for $A^{(\bm{s})}_{\mrm{id}}(\bm{a}\vert \bm{b})$:
\eq{
A^{(\bm{s})}_{\mrm{id}}(\bm{a}\vert \bm{b}) = \delta_{\bm{a},\bm{b}}.
\label{eq:identity_lem1}
}
\end{prop}

\begin{remark}
The identity~\eqref{eq:identity_lem1} involves the wave function of the inhomogeneous XXX spin chain; 
see \eqref{eq:A_P} for the definition of $A^{(\bm{s})}_P(\bm{a}\vert \bm{b})$.
For the homogeneous case where $s_j = 0$ for all $j$, 
the identity~\eqref{eq:identity_lem1} recovers the known result 
in the context of the symmetric simple exclusion process, 
whose Markov generator is equivalent to the homogeneous XXX Hamiltonian.
Indeed, in this case, the contour $\Gamma_{\bm{s}}$ becomes a sufficiently small circle enclosing $iu$ 
as shown in figure~\ref{fig:contour}~(b), and the identity \eqref{eq:identity_lem1} reads
\eqnn{
&\prod_{n=1}^M \left[\oint_{|\lambda_n - iu| \ll 1} \frac{d\lambda_n}{2\pi i}
\frac{2iu}{ (\lambda_n -iu)( \lambda_{n}  + iu) }
\right]
\\
&\times
\sum_{R\in S_M}
 \prod_{\substack{1\leq m <n \leq M \\  R^{-1}(n) < R^{-1}(m)}} \frac{\lambda_{m}-\lambda_n +2iu}{\lambda_m-\lambda_n-2iu} 
\prod_{n=1}^M  
\left[
\prod_{\alpha=1}^{\alpha_{R^{-1}(n)}-1} \frac{\lambda_{n}-iu}{\lambda_{n} +iu}\prod_{\beta=1}^{\beta_n-1} \frac{\lambda_n +iu}{\lambda_n  -iu}
\right]
=\prod_{j=1}^M\delta_{\alpha_j,\beta_j}.
}
Note that $\alpha_j$ and $\beta_j$ denote the positions of the $j$-th down-spin in the spin configurations $\bm{a}$ and $\bm{b}$, respectively.
Upon changing the integration variables to $\xi_j = \frac{\lambda_j-iu}{\lambda_j+iu}$ ($j=1,\dots,M$),
the above equation can be rewritten as
\eq{
 \prod_{n=1}^M \Big[ \oint_{|\xi_n|\ll 1} \frac{d\xi_n}{2\pi i \xi_n}\Big] \sum_{R\in S_M}
\prod_{\substack{1\leq m<n \leq M \\ R^{-1}(n)< R^{-1}(m)}} \Big[-\frac{1+\xi_m \xi_n -2\xi_n}{1 + \xi_m \xi_n -2\xi_m}\Big]
\prod_{n=1}^M \xi_{n}^{\alpha_{R^{-1}(n)}-\beta_n}= \prod_{j=1}^M \delta_{\alpha_j,\beta_j}.
\label{eq:tw}
} 
This identity is equivalent to the one derived by Tracy and Widom~\cite{Tracy_integral} for the asymmetric simple exclusion process, specifically for the symmetric case ($p=q$).
Thus, Proposition~\ref{prop} extends the known result~\eqref{eq:tw} to general inhomogeneities $\{s_j\}_{j=1}^N$.
\end{remark}

\noindent
\textbf{Proof of Proposition~\ref{prop}.}~
Proposition~\ref{prop} is equivalent to the following identity,
\eq{
\sum_{R\in S_M} I_s(R) = \prod_{j=1}^M \delta_{\alpha_j,\beta_j},
\label{eq:identity}
}
where
\eq{
I_s(R) &:= \prod_{n=1}^M \left[\oint_{\Gamma_{\bm{s}}} \frac{d\lambda_n}{2\pi i}
\frac{2iu}{ (\lambda_n -s_{\beta_n}-iu)( \lambda_{n} -s_{\alpha_{R^{-1}(n)}} + iu) }
\right]
 \prod_{\substack{1\leq m <n \leq M \\  R^{-1}(n) < R^{-1}(m)}} \frac{\lambda_{m}-\lambda_n +2iu}{\lambda_m-\lambda_n-2iu} 
\\
&\quad \times 
\prod_{n=1}^M  
\left[
\prod_{\alpha=1}^{\alpha_{R^{-1}(n)}-1} 
\frac{\lambda_{n} -s_{\alpha}-iu}{\lambda_{n} -s_{\alpha} +iu}
\prod_{\beta=1}^{\beta_n-1} \frac{\lambda_n-s_{\beta} +iu}{\lambda_n - s_\beta -iu}
\right].\label{eq:I_s}
}
We shall prove this identity~\eqref{eq:identity} by residue calculus and mathematical induction on $M$.

In the proof, we slightly deform the $\lambda$-contours in \eqref{eq:I_s} so that they become mutually distinct, denoting them by 
$\Gamma^{(k)}_{\bm{s}}$ ($k=1,\dots,M$).
This deformation guarantees that the factor in \eqref{eq:I_s},
\eq{
\prod_{\substack{1\leq m<n\leq M \\ R^{-1}(n) < R^{-1}(m)}} \frac{\lambda_m-\lambda_n +2iu}{\lambda_m-\lambda_n-2iu},
}
does not exhibit multiple poles with respect to any $\lambda_n$, thereby simplifying the residue calculations.
We remark that a similar deformation was employed in the proof of the identity by Tracy and Widom~\cite{Tracy_integral,Tracy_erratum}
(specifically in the erratum).
However, our proof strategy differs from theirs, which is based on a change of integration variables, 
since in the presence of inhomogeneities such a transformation does not appear to be available.

First, we consider the case $M=1$.
In this case, one has
\eq{
I_{s}(\mrm{id})
= \oint_{\Gamma^{(1)}_{\bm{s}}} \frac{d\lambda_1 }{2\pi i } \frac{2iu}{(\lambda_1 -s_{\beta_1} -iu)(\lambda_1-s_{\alpha_1}+iu)} 
\prod_{\alpha=1}^{\alpha_1-1}\frac{\lambda_1-s_\alpha-iu}{\lambda_1-s_\alpha+iu}\prod_{\beta=1}^{\beta_1-1} 
\frac{\lambda_1-s_\beta+iu}{\lambda_1-s_\beta -iu}
}
This integral is equal to zero when $\beta_1<\alpha_1$, since the integrand does not have poles at $\lambda_1=s_k+iu$ ($k=1,\dots,N$).
When $\beta_1=\alpha_1$, a simple calculation leads to $A_{\mrm{id}}(a_1\vert b_1)=1$. 
The above integral when $\alpha_1<\beta_1$ can be rewritten as
\eq{
\oint_{\Gamma^{(1)}_{\bm{s}}} \frac{d\lambda_1 }{2\pi i } \frac{2iu}{(\lambda_1 -s_{\beta_1} -iu)(\lambda_1-s_{\alpha_1}-iu)} 
\prod_{\beta=\alpha_1+1}^{\beta_1-1} \frac{\lambda_1-s_\beta+iu}{\lambda_1-s_\beta -iu}.\label{eq:M=1}
}
To evaluate this integral, we use the following form of the residue theorem.
For a rational function $f(z)$ and a counterclockwise contour $\mathcal{C}$,
the contour integral can be written as the negative sum of the residues
outside $\mathcal{C}$ on the extended complex plane,
\eq{
\oint_{\mathcal{C}} \frac{dz}{2\pi i} f(z)
= - \sum_{w \in S_{\mathrm{out}}} \Res_{z=w} f(z)- \Res_{z=\infty} f(z),
\label{eq:prop_formula}
}
where $S_{\mathrm{out}}$ denotes the set of poles of $f(z)$ lying outside
the contour $\mathcal{C}$.
From \eqref{eq:prop_formula}, it follows that the integral~\eqref{eq:M=1} is zero.
Thus we complete the proof of the identity \eqref{eq:identity} for $M=1$.

We now proceed to the case $M>1$, assuming that \eqref{eq:identity} holds for $M-1$.
For $n=1,\dots,M $, define the subset of $S_M$ as
\eq{
S^{(n)}_M := \{ R\in S_M | R(1) =n \}.
\label{eq:dec_sym}
}
The sets $S^{(n)}_M$ ($n=1,\dots,M$) are disjoint and their union gives the symmetric group,
\eq{
S_M = \bigsqcup_{n=1}^M S^{(n)}_M.
}
This decomposition allows us to evaluate the total sum over $S_M$ 
by separately considering contributions from each $S^{(n)}_M$.

Now consider $I_s(R)$ for $R\in S^{(1)}_M$. In this case, the $\lambda_1$-integral can be evaluated as in the case $M=1$.
After the $\lambda_1$-integration, we have
\eqnn{
I_s(R)  &= 
\delta_{\alpha_1,\beta_1}
\prod_{n=2}^{M} \left[\oint_{\Gamma^{(n)}_{\bm{s}}} 
\frac{d\lambda_n}{2\pi i}\frac{2iu}{ (\lambda_n -s_{\beta_n}-iu)( \lambda_{n} -s_{\alpha_{R^{-1}(n)}} + iu) }
\right]
\\
&\times
 \prod_{\substack{2\leq m <n \leq M \\  R^{-1}(n) < R^{-1}(m)}} \frac{\lambda_{m}-\lambda_n +2iu}{\lambda_m-\lambda_n -2iu} 
\times 
\prod_{n=2}^{M}  
\left[
\prod_{\alpha=1}^{\alpha_{R^{-1}(n)}-1} 
\frac{\lambda_{n} -s_{\alpha}-iu}{\lambda_{n} -s_{\alpha} +iu}
\prod_{\beta=1}^{\beta_n-1} \frac{\lambda_n-s_{\beta} +iu}{\lambda_n - s_\beta -iu}
\right].
}
Then, we obtain the following identity by taking the sum over all $R\in S^{(1)}_M$ and using the induction hypothesis,
\eq{
\sum_{R\in S^{(1)}_M }I_s(R ) = 
\prod_{j=1}^M \delta_{\alpha_j,\beta_j}.\label{eq:identity_1}
}

We next consider $I_s(R)$ for $R\in S^{(n)}_M$ ($n>1$).
In this case, we will show the following identity,
\eq{
\sum_{R\in S^{(n)}_M }I_s(R) = 0~~~~\text{for}~n>1.
\label{eq:identity_2}
}
When $ \beta_n \leq \alpha_1$, it follows that
\eq{
\beta_1 < \beta_n \leq \alpha_1 <\alpha_{R^{-1}(1)}.
}
Therefore, the integrand in \eqref{eq:I_s} is holomorphic inside the $\lambda_1$-contour, which 
proves \eqref{eq:identity_2}.

Henceforth, we assume $\alpha_1 <\beta_n$ and prove \eqref{eq:identity_2} in this case.
For $R\in S^{(n)}_M$, $I_s(R)$ can be expressed as
\eqnn{
I_s(R) &= \oint_{\Gamma^{(n)}_{\bm{s}}} \frac{d\lambda_n}{2\pi i } \frac{2iu}{(\lambda_n-s_{\beta_n}-iu)(\lambda_n-s_{\alpha_1}+iu)}
\prod_{m=1}^{n-1} \frac{\lambda_m -\lambda_n + 2 iu}{\lambda_m-\lambda_n - 2 iu} 
\prod_{\beta=\alpha_1}^{\beta_n-1}\frac{\lambda_n-s_\beta + iu}{\lambda_n-s_\beta-iu}
\\
&\times 
(\text{$\lambda_n$-independent part}).
}
Based on \eqref{eq:prop_formula}, 
we evaluate the $\lambda_n$-integral by picking up the residues at $\lambda_n = \lambda_m-2iu$ ($m=1,\dots,n-1$),
which leads to
\eqnn{
I_s(R) &=\sum_{m=1}^{n-1}
\prod_{\substack{k = 1\\ k \neq n }}^M \Big[\oint_{\Gamma^{(k)}_{\bm{s}}} \frac{d\lambda_k}{2\pi i}
\frac{2iu}{(\lambda_k-s_{\beta_k}-iu)(\lambda_k-s_{\alpha_{R^{-1}(k)}}+iu)} 
\Big]
\\
&\times
\frac{-8 u^2}{(\lambda_m-s_{\beta_n} - 3iu)(\lambda_m - s_{\alpha_1} -iu)} \prod_{\substack{l=1 \\ l\neq m}}^{n-1}
\frac{\lambda_l -\lambda_m + 4 iu}{\lambda_l -\lambda_m}
\prod_{\beta = \alpha_1}^{\beta_n-1} \frac{\lambda_m - s_\beta - iu}{\lambda_m - s_\beta - 3iu}
\\
&\times
\prod_{\substack{1<k<l\leq M \\ k,l\neq n,R^{-1}(k)>R^{-1}(l)}} \frac{\lambda_k -\lambda_l + 2iu}{\lambda_k-\lambda_l -2iu}
\times
\prod_{\substack{k=1 \\ k\neq n}}^M 
\Big[
\prod_{\alpha=1}^{\alpha_{R^{-1}(k)}-1} 
\frac{\lambda_{k} -s_{\alpha} - iu}{\lambda_{k} -s_{\alpha} + iu}
\prod_{\beta=1}^{\beta_k-1} \frac{\lambda_k -s_{\beta} + iu}{\lambda_k - s_\beta -iu}
\Big].
}

We next perform the integration with respect to $\lambda_m$. 
From the conditions $\alpha_1 < \alpha_{R^{-1}(m)}$ and $\beta_m< \beta_n$,
the integrand of the above expression has no singularities at $\lambda_m = s_j + iu$ ($j=1,\dots,N$). 
Furthermore, given our choice of the contour and the condition $s_j - s_k \neq 2iu$ for all $j, k$, 
no poles other than $\lambda_m = \lambda_l$ ($l=1, \dots, n-1; \, l \neq m$) can be located inside $\Gamma^{(m)}_{\bm{s}}$. 
Hence, the only possible singularities within the $\lambda_m$-contour are the points $\lambda_m = \lambda_l$.
Thus, after the $\lambda_m$-integration, one obtains
\eq{
I_s(R) = \sum_{m=1}^{n-1} \sum_{\substack{l=1 \\ l \neq m}}^{n-1} T_s (l,m;R)
}
where we define
\eq{
T_s(l,m;R) 
&:= 
\prod_{\substack{k=1 \\ k\neq n,m}}^M\Big[\oint_{\Gamma^{(k)}_{\bm{s}}} \frac{d\lambda_k}{2\pi i}
\frac{2iu}{(\lambda_k-s_{\beta_k}-iu)(\lambda_k-s_{\alpha_{R^{-1}(k)}}+iu)} 
\Big] 
\\
&\times
\frac{2iu\;  \chi[\lambda_l \in \mathrm{Int}~\Gamma_{\bm{s}}^{(m)} ] }{(\lambda_l - s_{\beta_m} -iu) (\lambda_l - s_{\alpha_{R^{-1}(m)}} + iu )}
\\
&\times
\Res_{\lambda_m = \lambda_l}
\Bigg[
\frac{-8 u^2}{(\lambda_m-s_{\beta_n} - 3iu)(\lambda_m - s_{\alpha_1} -iu)} \prod_{\substack{k=1 \\ k\neq m}}^{n-1}
\frac{\lambda_k -\lambda_m + 4 iu}{\lambda_k -\lambda_m}
\prod_{\beta = \alpha_1}^{\beta_n-1} \frac{\lambda_m - s_\beta - iu}{\lambda_m - s_\beta - 3iu}
\\
&\times
\prod_{\substack{1<j<k\leq M \\ j,k\neq n,R^{-1}(j)>R^{-1}(k)}} \frac{\lambda_j -\lambda_k + 2iu}{\lambda_j-\lambda_k -2iu}
\times
\prod_{\substack{ k = 1 \\ k\neq n} }^M
\Big[
\prod_{\alpha=1}^{\alpha_{R^{-1}(k)}-1} 
\frac{\lambda_{k} -s_{\alpha} - iu}{\lambda_{k} -s_{\alpha} + iu}
\prod_{\beta=1}^{\beta_k-1} \frac{\lambda_k -s_{\beta} + iu}{\lambda_k - s_\beta -iu}
\Big]
\Bigg]
\label{eq:T_s}
}
with the indicator function $\chi[\cdot]$ and the interior $\mathrm{Int}\; \Gamma^{(m)}_{\bm{s}}$ of the $\lambda_m$-contour.

Now we observe that
\eq{
T_s(l,m;R) + T_s(l,m;\Pi_{l,m} R) = 0.
\label{eq:T_spin}
}
This relation follows from the fact that, upon taking the residue at 
$\lambda_m = \lambda_l$, the first factor in the fourth line of 
\eqref{eq:T_s} for $\Pi_{l,m}R$ is the negative of that for $R$, 
whereas all the other factors coincide.
Therefore, using the above relation and the decomposition of $S_M^{(n)}$ by even and odd permutations as 
$S^{(n)}_M = S^{(n)}_{M,\mrm{even} }  \bigsqcup \Pi_{l,m} S^{(n)}_{M,\mrm{even}} $,
we obtain the desired identity~\eqref{eq:identity_2} as follows,
\eqnn{
\sum_{R\in S^{(n)}_M} I_s(R) &=\sum_{R\in S^{(n)}_M} \sum_{m=1}^{n-1} \sum_{\substack{l=1\\ l\neq m}}^{n-1} T_s(l,m;R)
\\
&= \sum_{m=1}^{n-1} \sum_{\substack{l=1\\ l\neq m}}^{n-1}\sum_{R\in S^{(n)}_{M,\mrm{even}}}[ T_s(l,m;R) + T_s(l,m;\Pi_{l,m} R)] 
\\
&= 0.
}

Finally, combining \eqref{eq:identity_1} and \eqref{eq:identity_2}, 
we obtain the identity \eqref{eq:identity}, which completes the proof of Proposition~\ref{prop} and thereby establishes Lemma~\ref{lem1}.


\subsection{Proof of Lemma~\ref{lem2}}\label{sub:charge}
We complete the proof of Step (b) by proving Lemma~\ref{lem2}. 

For notational convenience, we denote the summand on the left-hand side of \eqref{eq:identity_c_1} as $I_c(P)$.
Then the identity \eqref{eq:identity_c_1} can be written as
\eq{
\sum_{P\in S_N} I_c(P)
= \det[\delta_{x_j,y_k}\delta_{a_j,b_k}]_{j,k=1}^N .
\label{eq:identity_c}
}

As shown below, Lemma~\ref{lem2} can be proved in a manner similar to the proof of Proposition~\ref{prop}.
This is because the amplitude $A^{(\bm{s})}_P(\bm{a}\vert\bm{b})$ given in \eqref{eq:A_P_1} is composed of 
the $Y$-operators~\eqref{eq:yang},
\eq{
Y(s_j -s_k) = \frac{(s_j - s_k)\Pi + 2iu}{s_j - s_k + 2iu},
}
whose structure is analogous to that of the factors
$(\lambda_j - \lambda_k + 2iu)/(\lambda_j - \lambda_k - 2iu)$ appearing in \eqref{eq:I_s}.
In the present case, however, these factors are matrices acting on the auxiliary spin space, which makes the calculation more involved.

To handle this complication, it is useful to introduce the graphical representation of $A_P^{(\bm{s})}(\bm{a}\vert \bm{b})$.
The expression~\eqref{eq:A_P_1} admits a natural graphical representation,
which provides an intuitive understanding of the structure of the amplitude $A^{(\bm{s})}_P(\bm{a}\vert \bm{b})$.
Figure~\ref{fig:A_P} shows the corresponding diagrams for (a) $P=\Pi_{1,2}$ and (b) $P=\Pi_{1,3}$.
Each arrow is labeled with the rapidity $s_j$, and the positions 
where it starts and ends correspond to site $j$ and site $P^{-1}(j)$.
Each intersection indicates the position where the $Y$-operator acts and 
which pair of rapidities is scattered by the \( Y \)-operator.  
From the sequence and arrangement of these intersections, one can directly reconstruct 
the corresponding amplitude.
Note that the (braid) Yang--Baxter relation \eqref{eq:yang-baxter} guarantees the invariance of 
the amplitude under a change in the order of intersections.
\begin{figure}[htbp]
\begin{center}
\includegraphics[keepaspectratio, width = 12cm]{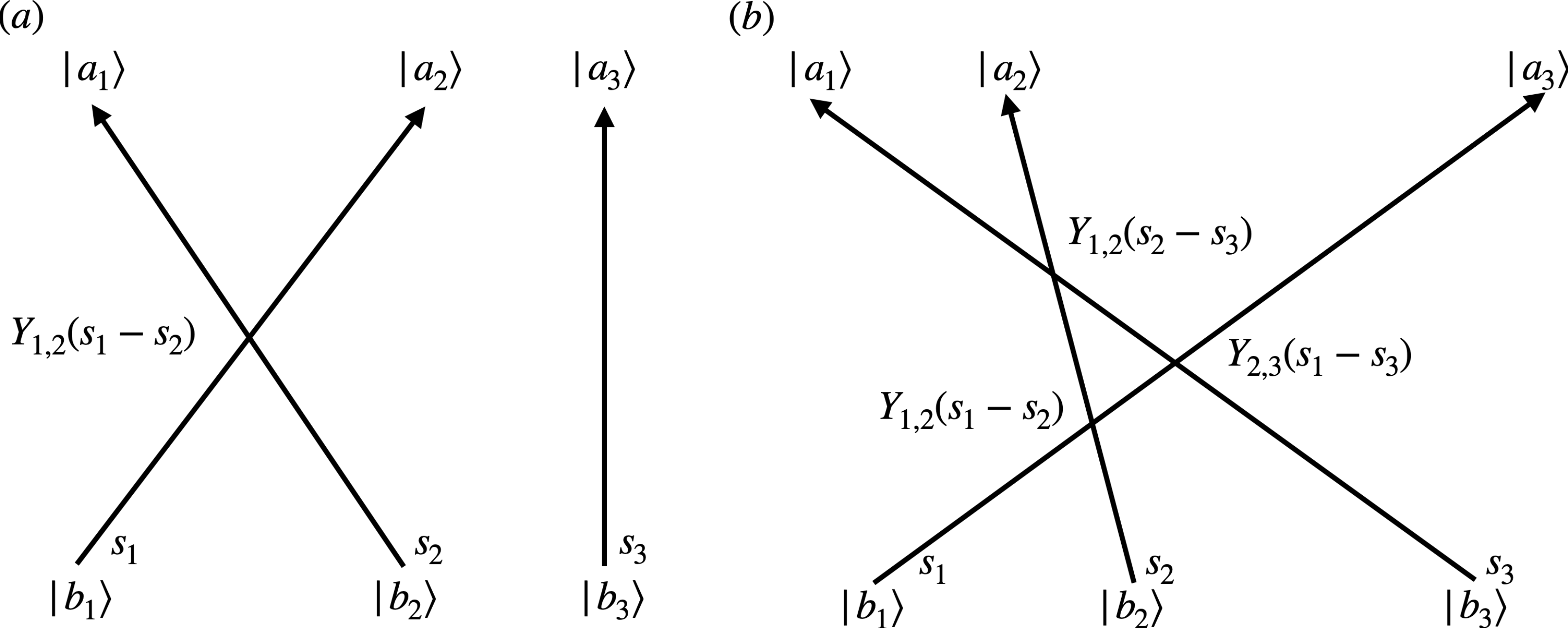}
\caption{Graphical representations of $A^{(\bm{s})}_P(\bm{a}\vert \bm{b})$ for (a) $P=\Pi_{1,2}$. (b) $P=\Pi_{1,3}$.
From Lemma~\ref{lem1}, one has
$A^{(\bm{s})}_{\Pi_{1,2}}(\bm{a}\vert \bm{b}) = \bra{\bm{a}} Y_{1,2}(s_1-s_2)\ket{\bm{b}}$
and 
$A^{(\bm{s})}_{\Pi_{1,3}}(\bm{a}\vert \bm{b}) = \bra{\bm{a}} Y_{1,2}(s_2-s_3) Y_{2,3}(s_1-s_3) Y_{1,2}(s_1-s_2)\ket{\bm{b}}$.
}
\label{fig:A_P}
\end{center}
\end{figure}

We proceed by induction on $N$.
The case $N=1$ is obvious. Let us consider $N=2$.
In this case, one has $I_c(\mrm{id})=\delta_{\bm{x},\bm{y}} \delta_{\bm{a},\bm{b}}$ straightforwardly. 
For $I_c(\Pi_{1,2})$, we make the substitution $z_2\to z_2/z_1$, which leads to
\eqnn{
I_c(\Pi_{1,2}) = -\oint_{|z_1|=r} \frac{dz_1}{2\pi z_1} \oint_{|z_2|=r} \frac{dz_2}{2\pi i z_2}
\bra{\bm{a}} \frac{(z_1^2-1 -z_2 +z_1^2/z_2)\Pi_{1,2}-4u z_1}{z_1^2-1 -z_2 +z_1^2/z_2-4 u z_1}\ket{\bm{b}}
z_2^{x_1-y_2} z_1^{x_2-x_1+y_2-y_1}.
}
Since $|z_2|=r$ on the $z_2$-contour, $1/(z_1^2 - 1 - z_2 + z_1^2 /z_2 -4u z_1)$ is analytic inside the $z_1$-contour.
The term $z_1^{x_2- x_1 + y_2-y_1}$ is also analytic from the conditions $x_1\leq x_2$ and $y_1\leq y_2$.
Therefore, from the residue theorem, it follows that
\eqs{
I_c(\Pi_{1,2})
&= -\oint_{|z_2|=r} \frac{dz_2}{2\pi i z_2} 
z_2^{x_1-y_2} \delta_{a_1,b_2}\delta_{a_2,b_1}  \delta_{x_1,x_2}\delta_{y_1,y_2}
\label{eq:z_1_z_2}
\\
&= -\delta_{x_1,y_2}\delta_{x_2,y_1}\delta_{y_1,y_2}\delta_{a_1,b_2}\delta_{a_2,b_1}
\\
&= -\delta_{x_1,y_2}\delta_{x_2,y_1}\delta_{a_1,b_2}\delta_{a_2,b_1}.
}
In the last line, we use the identity 
$\delta_{x_1,y_2}\delta_{x_2,y_1}\delta_{y_1,y_2}= \delta_{x_1,y_2}\delta_{x_2,y_1}$ 
which holds when $x_1\leq x_2$ and $y_1\leq y_2$.
Thus one obtains the identity~\eqref{eq:identity_c} for the case $N=2$.

We shall prove \eqref{eq:identity_c} for $N \geq 3$ assuming that it holds for $N - 1$ and $N - 2$.
As in the case of Proposition~\ref{prop}, we use the decomposition of the symmetric group:
$
S_N = \bigsqcup_{n=1}^N S^{(n)}_N.
$
The definition of $S^{(n)}_N$ is given in \eqref{eq:dec_sym}.
In what follows, we establish the following identities,
\eqs{
\sum_{P\in S^{(1)}_N \cup S^{(2)}_N }I_c(P) &= \det[\delta_{x_j,y_k}\delta_{a_j,b_k}]_{j,k=1}^N,
\label{eq:identity_G_1_2}
\\
\sum_{P\in S^{(n)}_N} I_c(P) &= 0~~~~\text{for}~n\geq3.
\label{eq:identity_G_n}
}
These two identities prove \eqref{eq:identity_c} for $N \geq 3$.

The following expression for the amplitude with $P \in S_N^{(n)}$ 
plays an important role in proving \eqref{eq:identity_G_1_2} and \eqref{eq:identity_G_n}:
\eq{
A_P^{(\bm{s})}(\bm{a}\vert \bm{b}) = \bra{\bm{a}}\cdots \prod_{1\leq m< n}^{\longrightarrow} Y_{m,m+1}(s_m-s_n)\ket{\bm{b}},
\label{eq:A_P_n}
}
where the omitted part $\cdots$ represents the sequence of the $Y$-operators independent of $z_n$,
and $\overset{\longrightarrow}{\prod}$ denotes the product ordered with site indices increasing from left to right.
We derive \eqref{eq:A_P_n} in appendix~\ref{app:details} using Lemma~\ref{lem1}.
We remark that this expression can also be understood from the graphical representation of $A_P^{(\bm{s})}(\bm{a}\vert \bm{b})$.

We first prove \eqref{eq:identity_G_1_2}. Now consider $P\in S^{(1)}_N$.
In this case, $A^{(\bm{s})}_P(\bm{a}\vert \bm{b})$ is independent of $z_1$, a fact that is evident from \eqref{eq:A_P_n}.
Consequently, the integrand of $I_c(P)$ factorizes into the product of $z_1^{x_1 - y_1-1}$ and a term independent of $z_1$.
This allows us to explicitly perform the $z_1$-integration.
By relabeling the remaining integration variables as $z_j \to z_{j-1}$ ($j = 2, \dots, N$) and 
defining $\tilde{P}\in S_{N-1}$ as $\tilde{P}(j):=P(j+1)-1$, ($j=1,\dots,N-1$), we obtain
\eqnn{
I_c(P)
= \delta_{x_1, y_1} \delta_{a_1, b_1}
\mrm{sgn}(\tilde{P})
\Big[\prod_{j=1}^{N-1}
\oint_{|z_j|=r^{N-j}} \frac{dz_j}{2\pi i z_j}
\Big]
A^{(\bm{s})}_{\tilde{P}}(a_2,\cdots,a_N \vert b_2,\cdots,b_N)
\prod_{j=1}^{N-1} z_{\tilde{P}(j)}^{x_{j+1}-y_{\tilde{P}(j)+1} }.
}
Then, by the induction hypothesis, we arrive at
\eq{
\sum_{P \in S^{(1)}_N} I_c(P) = \delta_{x_1,y_1}\delta_{a_1,b_1}\mrm{det}[\delta_{x_{j},y_k}\delta_{a_j,b_k}]_{j,k=2}^N.
\label{eq:G_1}
}

We next consider the case $P\in S^{(2)}_N$.
From \eqref{eq:A_P_n}, we have
\eqnn{
I_c(P)= \mrm{sgn}(P)
\Big[\prod_{j=1}^N
\oint_{|z_j|=r^{N-j}}\frac{dz_j}{2\pi z_j}
\Big]
\bra{\bm{a}}\cdots Y_{1,2}(s_1-s_2)\ket{\bm{b}}  z_2^{x_1-y_2} z_1^{x_{P^{-1}(1)}-y_1}
\prod_{j=3}^N z^{x_{P^{-1}(j)}-y_{j} }_{j}.
}
After the substitution $z_2\to z_2/z_1$, the above equation becomes
\eqnn{
I_c(P)&= \mrm{sgn}(P)\oint_{|z_1|= r^{N-1}}\frac{dz_1}{2\pi i}\oint_{|z_2|= r^{2N-3}} \frac{dz_2}{2\pi iz_2} \
\prod_{j=3}^N \oint_{|z_j|= r^{N-j}}\frac{dz_j}{2\pi iz_j}
\\
&\times 
\bra{\bm{a}}\cdots 
\frac{(z_1^2-1 -z_2 +z_1^2/z_2)\Pi_{1,2}-4u z_1}{z_1^2-1 -z_2 +z_1^2/z_2-4u z_1} \ket{\bm{b}}
\times  z_2^{x_1-y_2} z_1^{x_{P^{-1}(1)}-x_1+y_2-y_1-1}\prod_{j=3}^N z^{x_{P^{-1}(j)}-y_{j} }_{j}.
} 
Note that, from our choice of the contours, the factor $1/(z_1^2-1 -z_2 +z_1^2/z_2-4u z_1)$ and 
the omitted part $\cdots$ are holomorphic inside the $z_1$-contour.
In addition, when $P^{-1}(1) >2$, it follows that $x_{P^{-1}(1)}-x_1+y_2-y_1-1 \geq 0$.
Therefore, from the residue theorem, we have
\eq{
I_c(P)= 0~~\text{when}~P^{-1}(1) >2.
}
On the other hand, for the case $P^{-1}(1)=2$,  $I_c(P)$ can be evaluated as 
\eqnn{
I_c(P)&= -\mrm{sgn}(\tilde{P}) \delta_{x_1,y_2}\delta_{x_2,y_1}\delta_{a_1,b_2}\delta_{a_2,b_1}
\\
&\quad \times
\Big[\prod_{j=1}^{N-2}
\oint_{|z_j|=r^{N-j}} \frac{dz_j}{2\pi iz_j}
\Big]
A^{(\bm{s})}_{\tilde{P}}(a_3,\cdots,a_N \vert b_3,\cdots,b_N )
 \prod_{j=1}^{N-2} z_{\tilde{P}(j)}^{x_{j+2}-y_{\tilde{P}(j)+2}},
} 
where we define $\tilde{P}\in S_{N-2}$ as $\tilde{P}(j) := P(j+2)-2$, ($j=1,\dots,N-2$), and
used the identity $\delta_{x_1,y_2}\delta_{x_1,x_2}\delta_{y_1,y_2} = \delta_{x_1,y_2}\delta_{x_2,y_1}$ for $x_1\leq x_2$ and 
$y_1\leq y_2$.
Therefore, from the induction hypothesis, it follows that
\eq{
\sum_{P\in S_N^{(2)}}I(P) 
= -\delta_{x_1,y_2}\delta_{x_2,y_1}\delta_{a_1,b_2}\delta_{a_2,b_1} \det[\delta_{x_j,y_k}\delta_{a_j,b_k}]_{j,k=3}^N.
\label{eq:G_2}
} 

Combining \eqref{eq:G_1} and \eqref{eq:G_2}, we have
\eqnn{
\sum_{P\in S^{(1)}_N,S^{(2)}_N } I_c(P) = 
\delta_{x_1,y_1}\delta_{a_1,b_1}\det[\delta_{x_j,y_k}\delta_{a_j,b_k}]_{j,k=2}^N
-\delta_{x_1,y_2}\delta_{x_2,y_1}\delta_{a_1,b_2}\delta_{a_2,b_1} \det[\delta_{x_j,y_k}\delta_{a_j,b_k}]_{j,k=3}^N.
}
On the other hand, the determinant can be expanded as follows,
\eqsnn{
\det[\delta_{x_j,y_k}\delta_{a_j,b_k}]_{j,k=1}^N &= 
\sum_{P\in S_N } \mrm{sgn}(P)\delta_{x_1,y_{P(1)}} \cdots
\delta_{x_N,y_{P(N)}}\delta_{a_1,b_{P(1)}}\cdots\delta_{a_N,b_{P(N)}}
\\
&= \delta_{x_1,y_1} \delta_{a_1,b_1}\det[\delta_{x_j,y_k} \delta_{a_j,b_k}]_{j,k=2}^N 
-\delta_{x_1,y_2}\delta_{x_2,y_1}\delta_{a_1,b_2}\delta_{a_2,b_1} \det[\delta_{x_j,y_k}\delta_{a_j,b_k}]_{j,k=3}^N.
}
In the last line, we used the fact that $\prod_{j=1}^N \delta_{x_j,y_{P(j)}}\delta_{a_j,b_{P(j)}}$ vanishes
unless $P(1)=1$, or $P(1)=2$ and $P(2)=1$.
This restriction arises from the ordering
$x_1\leq \cdots \leq x_N$, $y_1\leq \cdots \leq y_N$, and the distinctness of the pairs $(x_j,a_j)$ and $(y_j,b_j)$.
Thus we obtain \eqref{eq:identity_G_1_2}.

We next prove \eqref{eq:identity_G_n}.
Let us consider $I_c(P)$ for $P\in S^{(n)}_N$, ($n\geq 3$).
When $x_2>y_1$, one can verify that $I_c(P)=0$ as the integrand is analytic inside the $z_1$-contour.
Hence, in the following, we assume $y_1-x_2\geq 0$.

We first evaluate the $z_n$-integral in $I_c(P)$.
After the substitution $z_n\to -1/z_n$, $I_c(P)$ becomes
\eqnn{
I_c(P) =(-1)^{x_1-y_n} \mrm{sgn}(P)\Big[\prod_{\substack{j=1 \\ j \neq n}}^N
\oint_{|z_j|=r^{N-j}} \frac{dz_j}{2\pi i z_j}\Big]
\oint_{|z_n|=1/r^{N-n}} \frac{dz_n}{2\pi i}
 A^{(\bm{s})}_P(\bm{a}\vert \bm{b})
 z_n^{y_n-x_1-1} 
 \prod_{\substack{j =1 \\ j \neq n}}^N z^{x_{P^{-1}(j)} -y_j}_{j}
} 
Note that since the function $s_n=(z_n-1/z_n)/2i$ is invariant under the substitution, $A^{(\bm{s})}_P(\bm{a}\vert \bm{b})$ also remains unchanged.
From the assumption $y_1-x_2\geq 0$, it follows that $y_n-x_1-1\geq 0$ for $n\geq 3$.
Hence the $z_n$-singularity of the integrand originates solely from $A^{(\bm{s})}_P(\bm{a}\vert \bm{b})$.
In \eqref{eq:A_P_n}, the $z_n$-dependent factor takes the following explicit form,
\eq{
Y_{m,m+1}(s_m-s_n)= -2i z_n\frac{2iu +(s_m-s_n)\Pi_{m,m+1}}{z_n^2-2i(s_m+2iu)z_n-1}.
}
The poles of this expression are located at $z_n = \rho(z_m)$ and $z_n=-1/\rho(z_m)$,
where $\rho(z)$ is the solution to $\rho^2 -2i(s(z)+2iu) \rho-1=0$ satisfying $\rho(z) \simeq z$ for $|z|\ll 1$.
The pole $z_n = \rho(z_m)$ lies inside the $z_n$-contour, while the other lies outside.
Hence, by performing the integration with respect to $z_n$, we obtain
\eq{
I_c(P) = (-1)^{x_1-y_n}\mrm{sgn}(P)
\sum_{m=1}^{n-1}
\Big[\prod_{\substack{j=1 \\ j \neq n,m}}^N \oint_{|z_j|=r^{N-j}} \frac{dz_j}{2\pi i z_j}\Big]
\oint_{|z_m|=r^{N-m}} \frac{dz_m}{2\pi i} \frac{2i}{\rho_m^2-1}
\\
\times
\Res_{s_n = s_m+2iu}[A^{(\bm{s})}_{P} (\bm{a} \vert \bm{b})]
 z_m^{x_{P^{-1}(m)}-y_m-1}\rho_m^{y_n-x_1+1} \prod_{\substack{j=1\\ j \neq n,m}}^N z_{j}^{x_{P^{-1}(j)}-y_j},
\label{eq:I_c_2}
}
where we defined the shorthand notation $\rho_m := \rho(z_m)$ and used the relation
\eq{
\Res_{z_n = \rho_m}Y_{m,m+1}(s_m-s_n) =  \frac{2i\rho_m^2}{\rho_m^2-1} \Res_{s_n = s_m+2iu} Y_{m,m+1}(s_m- s_n).
}

We next perform the $z_m$-integration in \eqref{eq:I_c_2}. 
Since $\rho_m = z_m + \order{z_m^2}$ for $|z_m|\ll1$, the following quantity,
\eq{
\frac{1}{\rho_m^2-1}z_m^{x_{P^{-1}(m)}-y_m-1} \rho_m^{y_n-x_1+1},
} 
is holomorphic inside the $z_m$-contour.
Hence again the singularity of the integrand originates solely from $\Res_{s_n = s_m+2iu}A^{(\bm{s})}_{P} (\bm{a} \vert \bm{b})$.
The following factor in this quantitiy may give rise to poles inside the $z_m$-contour,
\eqs{
&\frac{(s_l-s_m)\Pi +2iu}{s_l-s_m +2iu}\qquad(1\leq l < m~\text{and}~P^{-1}(m) < P^{-1}(l)),
\label{eq:(l,m)}
\\
&\frac{(s_l-s_n)\Pi+2iu}{s_l-s_n+2iu}\Big|_{s_n = s_m +2iu} = \frac{(s_l-s_m-2iu)\Pi +2iu}{s_l-s_m},\qquad(1 \leq l <m).
\label{eq:(l,n)}
}
See figure~\ref{fig:A_P_2} for the schematic illustration of these factors.
Note that, with our choice of the contours,
the following $z_m$-dependent factors are analytic inside the $z_m$-contour, 
\eq{
\frac{(s_m-s_j)\Pi+2iu}{s_m-s_j+2iu}\qquad(m<j).
} 
\begin{figure}[htbp]
\begin{center}
\includegraphics[keepaspectratio, width = 10cm]{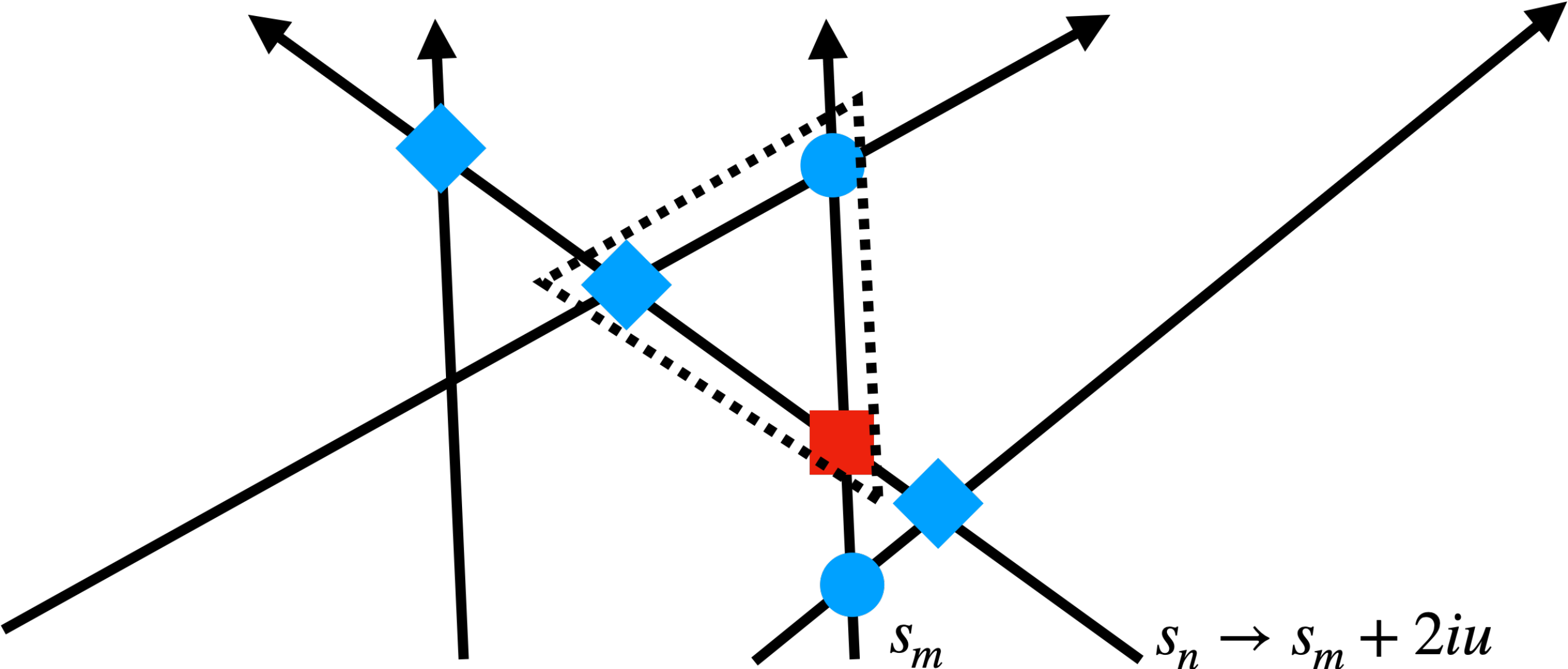}
\caption{Schematic illustration of $\Res_{s_n = s_m+2iu}[A^{(\bm{s})}_P (\bm{a} \vert \bm{b})]$.
The blue circles and diamonds represent the factors given in \eqref{eq:(l,m)} and 
\eqref{eq:(l,n)}, respectively, while
the red square corresponds to the pole-contributing factor in the $z_n$-integration.
The dashed triangle represents the factor given in \eqref{eq:triangle}.
}
\label{fig:A_P_2}
\end{center}
\end{figure}

Interestingly the denominator of \eqref{eq:(l,m)} cancels upon combination with other factors in 
$\Res_{s_n = s_m+2iu}A^{(\bm{s})}_P (\bm{a} \vert \bm{b})$.
In other words, we do not need to account for the corresponding poles in the integration.
The key point is that, by virtue of the Yang-Baxter relation \eqref{eq:yang-baxter}, 
one can choose a decomposition of the permutation $P$ into adjacent transpositions
such that the following sequence appears in 
$\Res_{s_n = s_m+2iu}A^{(\bm{s})}_P (\bm{a} \vert \bm{b})$:
\eqs{
&Y_{j,j+1}(s_l -s_m) Y_{j-1,j}(s_l-s_m-2iu)[-2iu(1-\Pi_{j,j+1})]\label{eq:triangle}
\\
&= \frac{(s_l-s_m)\Pi_{j,j+1} +2iu}{s_l-s_m +2iu}\frac{(s_l-s_m-2iu)\Pi_{j-1,j}+2iu}{s_l-s_m}[-2iu(1-\Pi_{j,j+1})]
\nonumber
\\
&=\frac{(s_l-s_m +2iu -2iu)\Pi_{j,j+1} + 2iu}{s_l-s_m+2iu}\frac{(s_l-s_m+2iu - 4iu)\Pi_{j-1,j} + 2iu}{s_l-s_m} [-2iu(1-\Pi_{j,j+1})]
\nonumber
}
where $j$ is some integer, and the remaining factors are regular at $s_m = s_l +2iu$.
See appendix~\ref{app:details} for an explicit example of such a decomposition.
This sequence corresponds to a triangle formed by the lines $l$, $m$, and $n$
in the graphical representation (see figure~\ref{fig:A_P_2}).
From the final expression above, one sees that if 
\eq{
(1-\Pi_{j,j+1})(2\Pi_{j-1,j}-1) (1-\Pi_{j,j+1}) =0,\label{eq:identity_perm}
}
then $\Res_{s_n = s_m+2iu}A^{(\bm{s})}_P (\bm{a} \vert \bm{b})$ is regular at $s_m = s_l +2iu$.
In fact, \eqref{eq:identity_perm} can be shown by the direct evaluation of its action on arbitrary basis vectors $\vert a_{j-1}, a_j, a_{j+1}\rangle$.
Consequently, in the $z_m$-integration in \eqref{eq:I_c_2}, the only contributing poles are those of \eqref{eq:(l,n)}.
Therefore, we obtain
\eq{
I_c(P) =  \sum_{m=1}^{n-1} \sum_{l=1}^{m-1} T_c(l,m;P),
\label{eq:I(P)}
}
where
\eqnn{
T_c(l,m;P) &:= (-1)^{x_1-y_n} \mrm{sgn}(P) 
\Big[\prod_{\substack{j =1 \\ j \neq n,m}}^N \oint_{|z_j|=r^{N-j} }\frac{dz_j}{2\pi i z_j} \Big]
\frac{2i}{\rho_l^2 -1} \frac{2i}{z_l^2 -1}
\\
&\times 
\Res_{s_m = s_l }\big[\Res_{s_n = s_m+2iu}[A^{(\bm{s})}_P (\bm{a} \vert \bm{b})]\big]
z_l^{x_{P^{-1}(l)} + x_{P^{-1}(m)} -y_l - y_m +1}\rho_l^{y_n-x_1+1}
\prod_{\substack{j=1\\ j \neq n,m,l}}^N z_j^{x_{P^{-1}(j)}-y_j}.
}

Now we observe that
\eq{
T_c(l,m;P) + T_c(l,m;\Pi_{l,m}P) = 0,\label{eq:T_charge}
}
which is analogous to \eqref{eq:T_spin}.
To prove this, it suffices to show that 
\eq{
\Res_{s_m = s_l }\big[\Res_{s_n = s_m+2iu}[A^{(\bm{s})}_P (\bm{a} \vert \bm{b})]\big] =
\Res_{s_m = s_l }\big[\Res_{s_n = s_m+2iu}[A^{(\bm{s})}_{\Pi_{l,m}P} (\bm{a} \vert \bm{b})]\big].
\label{eq:(n,m)(n,l)}
}
Although the validity of \eqref{eq:(n,m)(n,l)} can be seen from figure~\ref{fig:A_P_3}, 
we show it explicitly using the $Y$-operator in appendix~\ref{app:details}.
\begin{figure}[htbp]
\begin{center}
\includegraphics[keepaspectratio, width = 12cm]{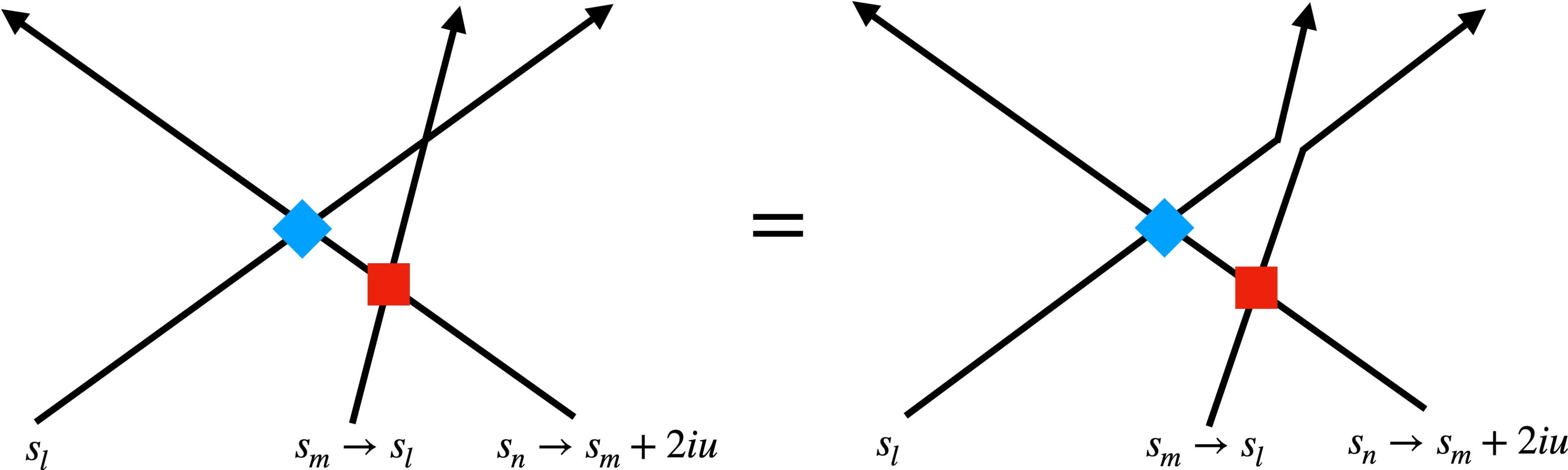}
\caption{Schematic illustration of \eqref{eq:(n,m)(n,l)}.
The red square and blue diamonds represent the pole-contributing factors in the $z_n$- and $z_m$-integrations, respectively.
The intersection between the lines $l$ and $m$ corresponds to the $Y$-operator $Y(s_l-s_m)$.
Note that this operator reduces to the identity operator when $s_l=s_m$.
}
\label{fig:A_P_3}
\end{center}
\end{figure}

Using \eqref{eq:T_charge} and the decomposition of $S^{(n)}_N$ by even and odd permutations as 
$S^{(n)}_N = S^{(n)}_{N,\mrm{even} }  \bigsqcup \Pi_{l,m} S^{(n)}_{N,\mrm{even}} $, 
we obtain
\eqnn{
\sum_{P \in S^{(n)}_N } I_c(P) &= \sum_{P\in S^{(n)}_N }\sum_{m=1}^{n-1}\sum_{l=1}^{m-1} T_c(l,m;P)
\\
& = \sum_{m=1}^{n-1}\sum_{l=1}^{m-1} \sum_{P\in S^{(n)}_{N,\mrm{even}} } [T_c(l,m; P) + T_c(l,m;\Pi_{l,m} P)]
\\
& = 0
}
Thus, we arrive at the identity~\eqref{eq:identity_G_n}, 
which, together with the identity~\eqref{eq:identity_G_1_2}, completes the proof of Lemma~\ref{lem2}.

\bigskip
Finally, from Lemma~\ref{lem1} and Lemma~\ref{lem2}, we establish the proof of Step (b), thereby completing the proof of Theorem~\ref{thm}.

\section{Conclusion}\label{sec:conclusion}
In this work, we have presented an exact integral formula for the multi-particle 
propagator of the one-dimensional Fermi--Hubbard model on the infinite lattice. 
Through this formula, the time dependence of arbitrary finite-particle wave functions 
can be computed explicitly in the Hubbard model. 
The formula can also be applied to open quantum models related to the Hubbard model~\cite{Medvedyeva_exact,Nakagawa_2021,Yoshida_liouvillian,Marche_universality}. 
Thus, it establishes a microscopic starting point for the exact analysis 
of nonequilibrium dynamics in these systems.

The proof of the formula has been established by verifying two conditions: 
the satisfaction of the Schr\"{o}dinger equation and the fulfillment of the initial condition.
While the former follows directly from the nested Bethe ansatz~\cite{Lieb_hubbard,Essler_hubbard}, the latter presents a non-trivial challenge.
We resolved this by leveraging the nested structure of the Bethe wave function; 
specifically, we performed successive residue calculations with respect to both spin and charge rapidities.
Crucially, our derivation works directly on the infinite lattice, thereby avoiding reliance on the string hypothesis.

The integral formula for general $N$ and $M$ presented in this work is expected to serve as a versatile starting point 
for investigating nonequilibrium dynamics in the Hubbard model and related systems. 
Combined with appropriate physical quantities, it could yield tractable expressions even for many-body observables. 
In particular, for the tight-binding chain with dephasing noise---a related open quantum model---we have obtained 
a simple Fredholm determinant formula for the full counting statistics of the current 
by applying the formula to the cases $N=2n$ and $M=n$ ($n \in \mathbb{N}$), 
which will be reported in a forthcoming paper. 
We also expect that, in the Hubbard model itself with finite-density initial states, 
our integral formula may be applicable to relatively simple observables such as return probabilities.
Indeed, in the XXZ spin chain, 
several studies have attempted to investigate the nonequilibrium dynamics of finite-density initial states
by emplyoing an integral representation of the propagator~\cite{Saenz2022, Fujimoto_quantum}.
Our results should provide a basis for future studies in a similar direction accessing the finite-density dynamics of the Hubbard model.
Another natural direction is to extend the present approach to the Maassarani model~\cite{Maassarani_1998}, which is the SU($n$) generalization of the Hubbard model.
Our framework may provide a useful starting point for analyzing nonequilibrium dynamics 
in systems with higher-internal symmetries.

\ack
TI, KF, and TS are grateful to Fabian Essler, Cristian Giardin\`a, Takashi Imamura, Matteo Mucciconi for helpful discussions and comments.
The work of TI has been supported by JST SPRING, Japan Grant Number JPMJSP2180.
The work of KF has been supported by JSPS KAKENHI Grant No. JP23K13029.
The work of TS has been supported by JSPS KAKENHI Grants No. JP21H04432, No. JP22H01143, and No. JP23K22414.

\appendix
\section{Open quantum models related to the Hubbard model}\label{app:open}
We explain open quantum models relevant to the Hubbard model with complex interaction strength and 
the application of the propagator \eqref{eq:def_pro} to these models.
In this appendix, we consider open quantum systems whose time evolution is governed by 
the Gorini-Kossakowski-Sudarshan-Lindblad (GKSL) equation~\cite{Lindblad,Gorini}:
\begin{equation}
\frac{d}{dt} \hat{\rho}(t) =\mathcal{L}[\hat{\rho}(t)]:= -i [\hat{H},\hat{\rho}(t)] + \sum_{x}\left[ 2\hat{L}_x \hat{\rho}(t) \hat{L}_x^\dagger 
-\{\hat{L}^\dagger_x \hat{L}_x,\hat{\rho}(t)\}\right].
\label{eq:GKSL}
\end{equation}
Here, $\mathcal{L}$ is the superoperator referred to as the Liouvillian, $\hat{H}$ denotes the system's Hamiltonian, and 
the Lindblad operator $\hat{L}_x$ represents the dissipation caused by coupling to the environment.
For a comprehensive introduction to open quantum systems, we refer the reader to~\cite{Breuer}.

\subsection{Tight-binding chain with dephasing noise}
A tight-binding chain with dephasing noise is described by the GKSL equation~\eqref{eq:GKSL}
with the one-dimensional tight-binding Hamiltonian
\eq{
\hat{H} := -\sum_{x\in \mathbb{Z}} (\hat{a}^\dagger_x \hat{a}_{x+1}+ \hat{a}^\dagger_{x+1} \hat{a}_x) 
}
and the Lindblad operator representing dephasing noise 
\eq{
\hat{L}_x := \sqrt{2\gamma_{\mrm{dep}} } \hat{a}^\dagger_x \hat{a}_x.
}

Under this setting, as shown in \cite{Medvedyeva_exact}, 
the Liouvillian $\mathcal{L}$ can be mapped to the Hubbard Hamiltonian with imaginary interaction.
To show this, define superoperators, 
\eqs{
&\mathcal{C}_{x,\downarrow}[\hat{\rho}] := \hat{a}_x \hat{\rho},
\\
&\mathcal{C}_{x,\uparrow}[\hat{\rho}] := \Big[\prod_{x\in \mathbb{Z}} e^{i\pi \hat{a}^\dagger_x \hat{a}_x}\Big]\hat{\rho} 
\hat{a}_x^\dagger,
}
and the Hilbert-Schmidt inner product for operators $\hat{A}$ and $\hat{B}$,
\eq{
(\hat{A},\hat{B}) := \mathrm{Tr}[\hat{A}^\dagger \hat{B}].
}
Then, the adjoint superoperators of $\mathcal{C}_{x,\uparrow}$ and $\mathcal{C}_{x,\downarrow}$ are given by
\eqs{
&\mathcal{C}^\dagger_{x,\downarrow}[\hat{\rho}] = \hat{a}^\dagger_x  \hat{\rho},
\\
&\mathcal{C}^\dagger_{x,\uparrow}[\hat{\rho}] 
= \Big[\prod_{x\in \mathbb{Z}} e^{i\pi \hat{a}^\dagger_x \hat{a}_x}\Big]\hat{\rho} \hat{a}_x.
}
These superoperators satisfy the canonical anti-commutation relations,
\eq{
\{\mathcal{C}_{x,\sigma}, \mathcal{C}_{y,\sigma'}\} = 0, \quad
\{\mathcal{C}_{x,\sigma}, \mathcal{C}_{y,\sigma'}^\dagger\} = \delta_{x,y}\delta_{\sigma,\sigma'}.
}
Therefore they can be regarded as the creation and annihilation operators of spin-1/2 fermions.
The vacuum state annihilated by all $\mathcal{C}_{x,\sigma}$ 
corresponds to the projector onto the original vacuum:
\eq{
\hat{\Omega} := \ket{0}\bra{0}.
}

In terms of the above superoperators, the Liouvillian is equivalent to the Hubbard Hamiltonian \eqref{eq:hubbard}
with imaginary interaction:
\eq{
i\mathcal{U}\mathcal{L}\mathcal{U}^\dagger = -\sum_{x\in \mathbb{Z},\sigma \in \{\uparrow,\downarrow\}} 
(\mathcal{C}^\dagger_{x,\sigma} \mathcal{C}_{x+1,\sigma} + \mathcal{C}^\dagger_{x+1,\sigma} \mathcal{C}_{x,\sigma}) 
+ 2i\gamma_{\mrm{dep}} \sum_{x\in \mathbb{Z}} (2\mathcal{C}^\dagger_{x_,\uparrow} \mathcal{C}_{x,\uparrow} 
\mathcal{C}_{x,\downarrow}^\dagger  \mathcal{C}_{x,\downarrow} - \mathcal{C}^\dagger_{x_,\uparrow} \mathcal{C}_{x,\uparrow} 
-\mathcal{C}_{x,\downarrow}^\dagger  \mathcal{C}_{x,\downarrow} ),
}
where we define the unitary superoperator as
\eq{
\mathcal{U} := \prod_{x\in \mathbb{Z}} e^{i\pi \mathcal{C}^\dagger_{2x-1,\downarrow}\mathcal{C}_{2x-1,\downarrow}}.
}
From this correspondence, the time evolution of a density matrix element is equivalent to
the propagator of the Hubbard model with imaginary interaction:
\eq{
\bra{x_1,\cdots,x_M} e^{\mathcal{L}t } [\ket{y_{M+1},\cdots,y_N}\bra{y_{1},\cdots,y_M }] \ket{x_{M+1},\cdots, x_N}
\\
= \Big[\prod_{j=1}^M (-1)^{x_j-y_j}\Big] e^{-2\gamma_{\mrm{dep}} N }\psi_t (\bm{x};\bm{a}_0\vert \bm{y};\bm{a}_0),
}
where 
\eq{
\bm{a}_0 := (\underbrace{\downarrow,\cdots,\downarrow}_M,\underbrace{\uparrow,\cdots,\uparrow}_{N-M}).
}
Thus once the exact expression of the propagator is obtained,
one can exactly calculate
the time dependence of the density matrix for finite-particles in a tight-binding chain with dephasing noise.
Furthermore, as shown in \cite{Medvedyeva_exact}, correlation functions satisfy the same equation of motion as
the density matrix elements.
Therefore, physical quantities for infinitely many particles can be calculated from the propagator of finite particles.

We remark that a non-unitary circuit, which provides an integrable trotterization of this model, has been proposed in \cite{Sa_2021}.

\subsection{Hubbard model under two-body loss}
The Hubbard model subject to two-body loss is described by the GKSL equation~\eqref{eq:GKSL}
with the one-dimensional Fermi-Hubbard Hamiltonian \eqref{eq:Hubbard}
\eq{
\hat{H} = - \sum_{x\in \mathbb{Z}}\sum_{\sigma= \downarrow,\uparrow} \left(\hat{c}^\dagger_{x,\sigma} \hat{c}_{x+1,\sigma} 
+ \hat{c}^\dagger_{x+1,\sigma}\hat{c}_{x,\sigma}\right) + 4u
\sum_{x\in \mathbb{Z}} \hat{n}_{x,\downarrow} \hat{n}_{x,\uparrow}
}
and the Lindblad operators representing the two-body loss
\eq{
\hat{L}_x := \sqrt{4\gamma_{\mrm{loss}}} \hat{c}_{x,\uparrow}\hat{c}_{x,\downarrow}.
}
As shown in~\cite{Nakagawa_2021}, the eigenvalues and eigenstates of this model are determined by the Hubbard model
with complex interaction,
\eqs{
\hat{H}_{\mathrm{eff}} &:= \hat{H} -i  \sum_{x\in \mathbb{Z}} \hat{L}_x ^\dagger \hat{L}_x 
\\
&= -\sum_{x\in \mathbb{Z}}\sum_{\sigma= \downarrow,\uparrow} \left(\hat{c}^\dagger_{x,\sigma} \hat{c}_{x+1,\sigma} 
+ \hat{c}^\dagger_{x+1,\sigma}\hat{c}_{x,\sigma}\right) + 4(u-i\gamma_{\mrm{loss}})
\sum_{x\in \mathbb{Z}} \hat{n}_{x,\downarrow} \hat{n}_{x,\uparrow}.
\label{app:hubbard}
}
Furthermore, the probability that $N$ particles initially located 
at sites $\bm{y}=(y_1,\cdots,y_N)$ with spin configurations $\bm{b} =(b_1,\cdots,b_N) $ are found 
at sites $\bm{x}= (x_1,\cdots,x_N)$ with spin configurations $\bm{a} = (a_1,\cdots,a_N)$ at time $t$ is given by
the propagator of the Hubbard model with the complex interaction~\eqref{app:hubbard},
\eq{
\mathrm{Pr}[\bm{x};\bm{a}\vert \bm{y};\bm{b}] = |\psi_t(\bm{x};\bm{a}\vert \bm{y};\bm{b})|^2.
}
This fact can be derived from the decomposition of the Liouvillian into two parts, $\mathcal{L}= \mathcal{K}+\mathcal{A}$
~\cite{Yoshida_liouvillian, Marche_universality}:
\eq{
&\mathcal{K}[\hat{\rho}] :=  -i \hat{H}_{\mrm{eff}} \hat{\rho} +i \hat{\rho} \hat{H}^\dagger_{\mrm{eff}},
\\
&\mathcal{A}[\hat{\rho}] := 8\gamma_{\mrm{loss}}  \sum_{x\in \mathbb{Z}} 
\hat{c}_{x,\uparrow}\hat{c}_{x,\downarrow}\hat{\rho} \hat{c}^\dagger_{x,\downarrow}\hat{c}^\dagger_{x,\uparrow}.
}
Note that $\mathcal{K}$ conserves the total number of particles, whereas $\mathcal{A}$ decreases it.

\section{Inhomogeneous XXX structure of the scattering amplitude }\label{app:yang}
In this appendix, we describe the relation between the scattering amplitude~\eqref{eq:spin_wave} 
and the Bethe ansatz wave function of the inhomogeneous XXX spin chain, 
and derive \eqref{eq:sca} based on this relation.

As explained in section \ref{sec:bethe}, 
the scattering amplitude for charges given in \eqref{eq:spin_wave} is the Bethe ansatz wave function 
of the inhomogeneous XXX spin chain~\cite{Essler_hubbard,Korepin_1993}.
To show this, we first introduce the $L$-matrix at site $n$ ($n=1,\dots,N$) as
\eq{
L_n (\lambda-\mu) := 
\frac{1}{\lambda-\mu +iu }
\begin{pmatrix}
\lambda-\mu + iu \sigma_n^z & 2iu \sigma_n^-
\\
2iu \sigma_n^+ & \lambda-\mu - iu \sigma_n^z
\end{pmatrix},
}
where $\sigma_n^{+}$, $\sigma_n^{-}$, and $\sigma^z_n$ denote the spin-1/2 raising, lowering, and $z$-component Pauli matrices at site $n$, respectively.
The monodromy matrix of the inhomogeneous XXX model is then defined by
\eq{
T(\lambda\vert \bm{s}P):= L_N(\lambda-s_{P(N)})\cdots L_1(\lambda-s_{P(1)}).
}
Note that $s_j$ is given by \eqref{eq:s_j}, and it represents the inhomogeneity in the above equation.

To gain an intuitive understanding of the $L$-matrix and the monodromy matrix, 
it is helpful to introduce their graphical representations.
The $L$-matrix is depicted as the crossed lines shown in figure~\ref{fig:L_monodromy}~(a).
In this figure, the vertical line with the rapidity $\mu$ corresponds to the quantum space, describing the spin-1/2 degree of freedom at site $n$
whereas the horizontal line with the rapidity $\lambda$ represents the auxiliary space.
The monodromy matrix is represented by a single horizontal line intersecting $N$ vertical lines, 
as illustrated in figure~\ref{fig:L_monodromy}~(b), 
since the monodromy matrix is defined as an ordered product of $L$-matrices.
\begin{figure}[htbp]
\begin{center}
\includegraphics[keepaspectratio, width = 12cm]{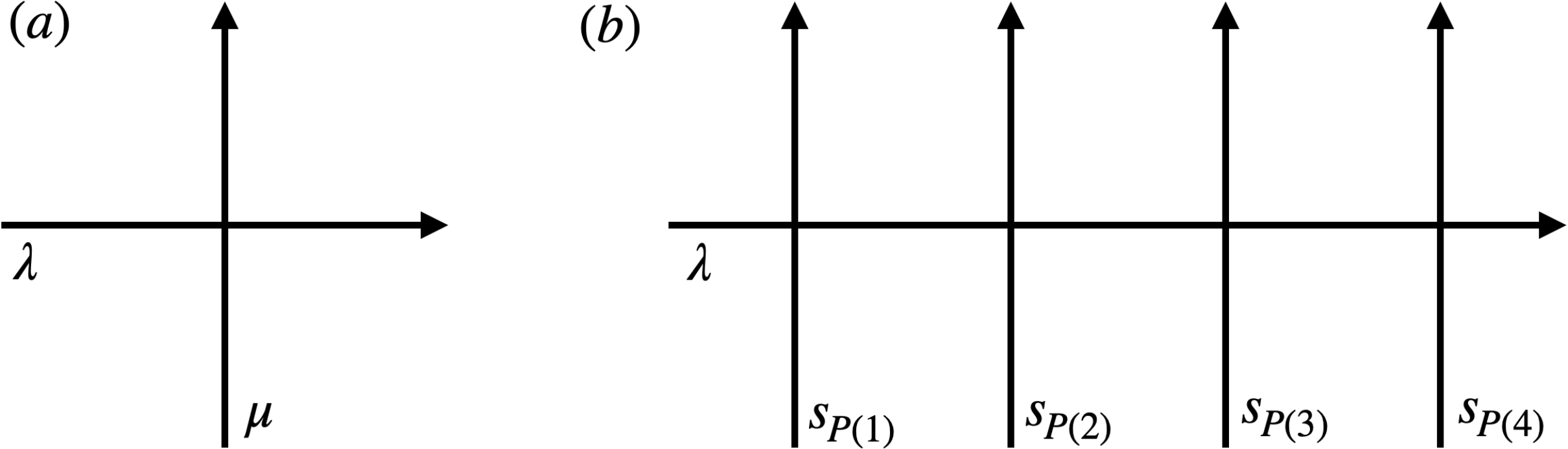}
\caption{Graphical representations of the $L$-matrix and monodromy matrix.
(a) $L$-matrix $L_n(\lambda-\mu)$, (b) Monodromy matrix $T(\lambda\vert \bm{z}P)$ for the case of $N=4$.}
\label{fig:L_monodromy}
\end{center}
\end{figure}

The $B$-operator is given by the (1,2)-entry of the monodromy matrix:
\eq{
B(\lambda\vert \bm{s}P) := T(\lambda \vert \bm{s}P)_{1,2}.
}
Acting on the all-up state with these $B$-operators yields the Bethe ansatz state of the inhomogeneous XXX model,
\eq{
B(\lambda_1\vert \bm{s}P)\cdots B(\lambda_M\vert \bm{s}P)\ket{\uparrow,\cdots,\uparrow}.
}
Then one sees that this state is equivalent to \eqref{eq:spin_wave}:
\eq{
\ket{\bm{s}P;\bm{\lambda}} = 
\prod_{1\leq m< n\leq M} \Big[\frac{\lambda_m-\lambda_n}{\lambda_m-\lambda_n -2iu} \Big]
B(\lambda_1\vert \bm{s}P)\cdots B(\lambda_M\vert \bm{s}P)\ket{\uparrow,\cdots,\uparrow}.
\label{eq:spin_wave_B}
}
This can be derived by considering the so-called ``multisite model''; see~\cite{Essler_hubbard,Korepin_1993}.
\begin{figure}[htbp]
\begin{center}
\includegraphics[keepaspectratio, width = 12cm]{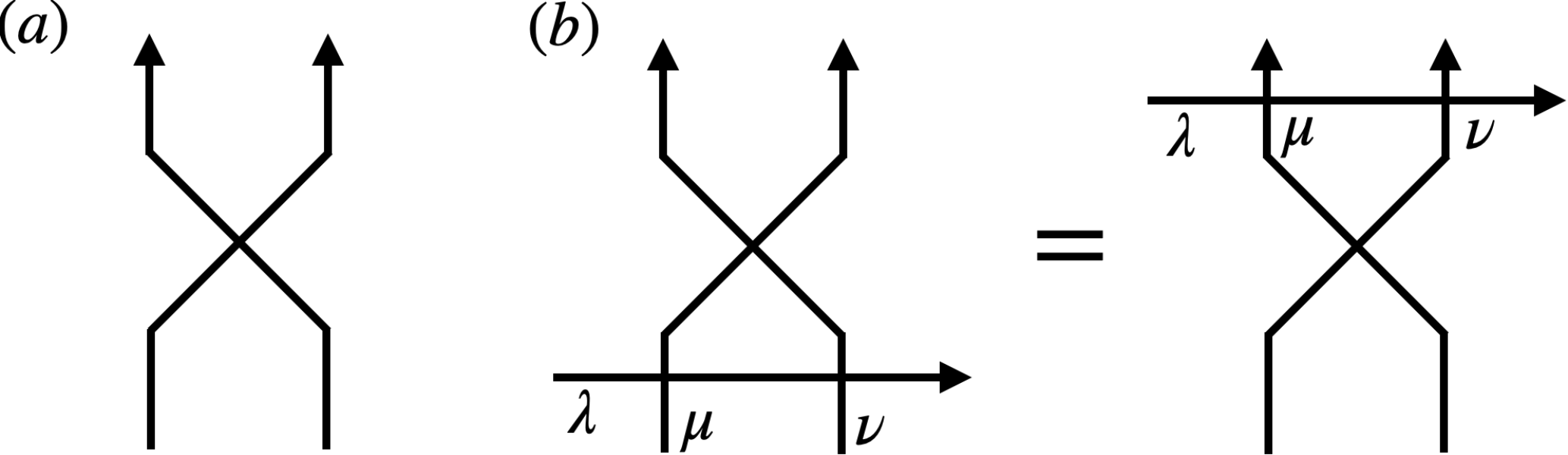}
\caption{Graphical representations of (a) the $Y$-operator and (b) the Yang-Baxter relation for the $Y$-operator and the $L$-matrix.}
\label{fig:yang}
\end{center}
\end{figure}

The $Y$-operator can be regarded as the $\check{R}$-matrix acting on the quantum space rather than on the auxiliary space as depicted in
figure~\ref{fig:yang}~(a),
and it satisfies the Yang-Baxter relation with the $L$-matrix,
\eq{
Y_{n,n+1}(\mu-\nu) L_{n+1}(\lambda -\nu) L_n(\lambda -\mu)  
= 
L_{n+1}(\lambda - \mu) L_n(\lambda-\nu)   Y_{n,n+1}(\mu-\nu).
}
See figure~\ref{fig:yang}~(b) for the graphical representation of the Yang-Baxter relation.
Then, using this relation together with \eqref{eq:spin_wave_B}, we can make the following calculation,
\eqnn{
Y_{n,n+1}(s_{P(n)} - s_{P(n+1)})\ket{\bm{s}P;\bm{\lambda}} &=~ \vcenteredinclude{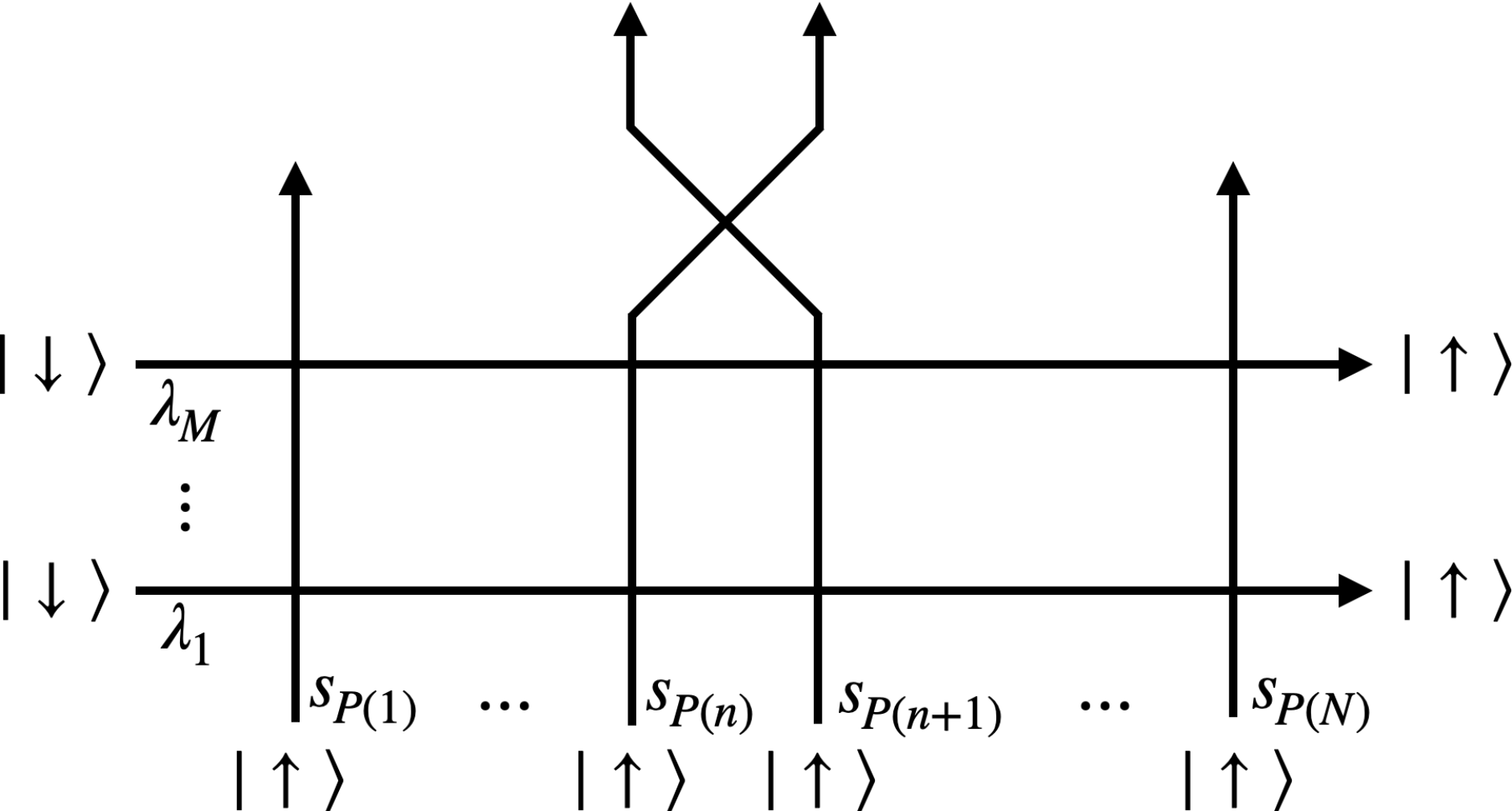}
\\
&=~\vcenteredinclude{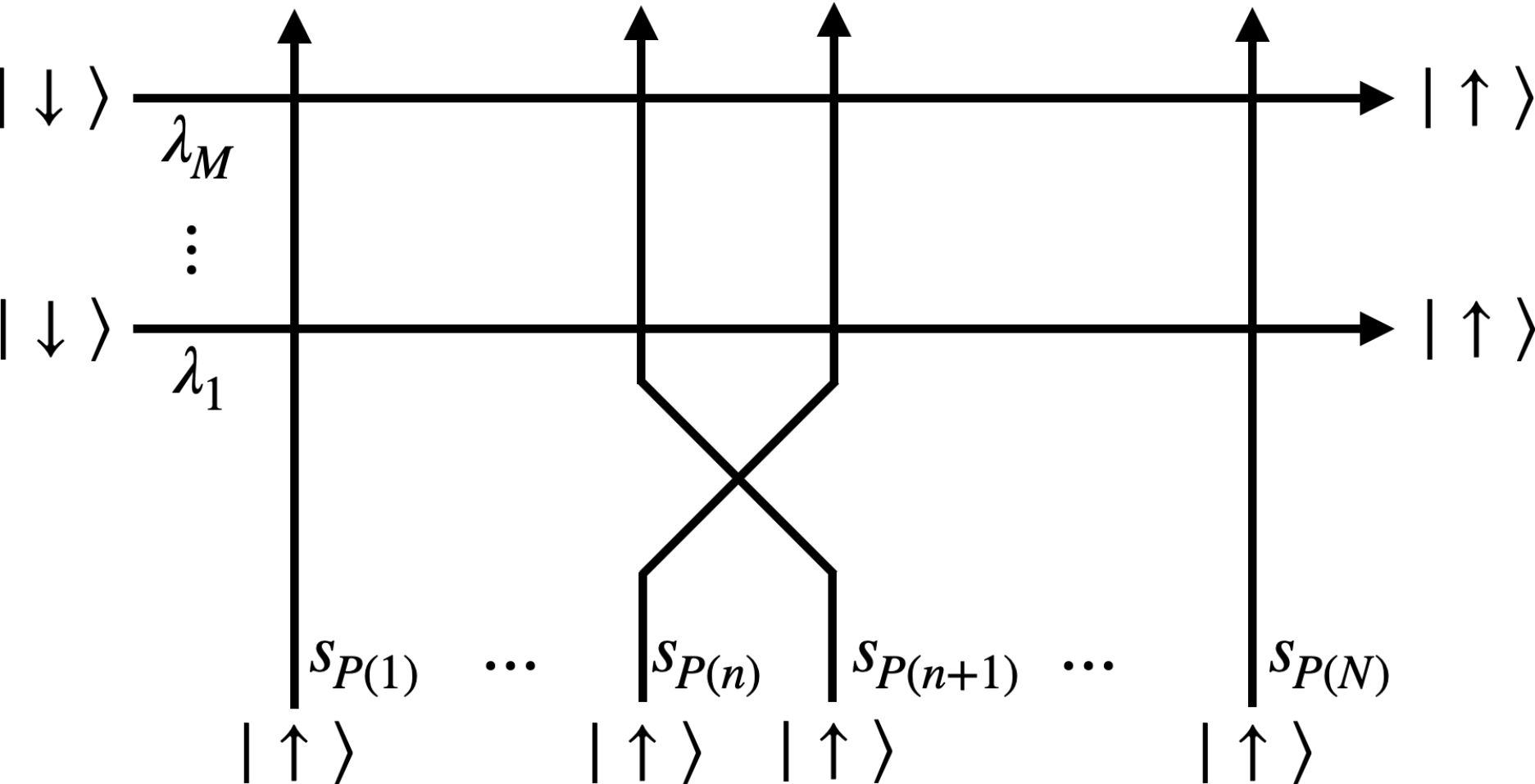}
\\
&= \ket{\bm{s}P\Pi_{n,n+1};\bm{\lambda}},
}
where we used the relation $Y_{n,n+1}(\mu) \ket{\uparrow,\cdots,\uparrow} =\ket{\uparrow,\cdots,\uparrow} $ in the last line.
Thus we complete the derivation of \eqref{eq:sca}.

\section{ Technical details in the proof of Lemma~\ref{lem2}}\label{app:details}
Here we derive \eqref{eq:A_P_n}, \eqref{eq:triangle}, and \eqref{eq:(n,m)(n,l)}. 
In the derivation, we use the following property of $A^{(\bm{s})}_P(\bm{a}\vert \bm{b})$,
\eq{
A_{QR}(\bm{a}\vert \bm{b}) = \sum_{\bm{a}'\in \{\uparrow,\downarrow\}^N } A^{(\bm{s}')}_{R}(\bm{a}\vert \bm{a}')A_Q^{(\bm{s})}(\bm{a}'\vert \bm{b}),
\label{eq:A_P_product}
}
where $s'_j := s_{Q(j)}$ for $j=1,\ldots,N$.
This relation can be proved directly from Lemma~\ref{lem1}.

\subsection{Derivation of \eqref{eq:A_P_n}}
We show the following expression for the amplitude with $P = S_N^{(n)}$,
\eq{
A_P^{(\bm{s})}(\bm{a}\vert \bm{b}) = \bra{\bm{a}}\cdots \prod_{1\leq m< n}^{\longrightarrow} Y_{m,m+1}(s_m-s_n)\ket{\bm{b}},
\label{eq:A_P_n_app}
}
where the omitted part $\cdots$ represents the sequence of the $Y$-operators independent of $z_n$,
and $\overset{\longrightarrow}{\prod}$ denotes the product ordered with site indices increasing from left to right.

Let us decompose $P\in S_N^{(n)}$ as $P = U_{1,n} Q$ with 
$U_{1,n}:= \overset{\longleftarrow}{\prod}_{1\leq m < n} \Pi_{m,m+1}$ and some $Q$.
Then using \eqref{eq:A_P_product}, we have
\eqsnn{
A_P^{(\bm{s})}(\bm{a}\vert \bm{b}) 
&= \sum_{\bm{a}' \in \{\uparrow,\downarrow\}^N } A^{(\bm{s}')}_Q(\bm{a}\vert \bm{a}')
A^{(\bm{s})}_{U_{1,n}} (\bm{a}'\vert \bm{b})
\\
&=\sum_{\bm{a}' \in \{\uparrow,\downarrow\}^N } A^{(\bm{s}')}_Q(\bm{a}\vert \bm{a}')
\bra{\bm{a}'} \prod^{\longrightarrow}_{1\leq m < n} Y_{m,m+1} (s_m - s_n) \ket{\bm{b}}.
}
In the second line, we used \eqref{eq:A_P_1}. Note that $s'_1 =s_n $ and $Q(1) = 1$.
Hence $A^{(\bm{s}')}_Q(\bm{a}\vert \bm{a}')$ is independent of $z_n$.
Thus we obtain \eqref{eq:A_P_n_app}.

\subsection{Derivation of \eqref{eq:triangle}}
We show that for $P\in S_N^{(n)}$ and $1\leq l < m < n$ satisfying $P^{-1}(m)<P^{-1}(l)$,
the factors in the residue $\Res_{s_n = s_m+2iu} A^{(\bm{s})}_P(\bm{a}\vert \bm{b})$
can be arranged so that the following sequence appears
\eq{
Y_{j,j+1}(s_l-s_m)Y_{j-1,j}(s_l-s_m-2iu) [-2iu(1-\Pi_{j,j+1})],
\label{eq:triangle_1}
}
where $j$ is some integer and the remaining factors are regular at $s_m = s_l +2iu$.

Let us set $\mu = P^{-1}(l)$ and $\nu = P^{-1}(m)$ for notational convenience.
We decompose $P\in S_N^{(n)}$ as $P = Q\Pi_{P^{-1}(n),P^{-1}(l)}=Q\Pi_{1,\mu}$ with some permutation $Q$.
Using \eqref{eq:A_P_product}, we obtain
\eq{
A^{(\bm{s})}_P(\bm{a}\vert \bm{b}) = \sum_{\bm{a}' \in \{\uparrow,\downarrow\}^N } 
A^{(\bm{s}')}_{\Pi_{1,\mu}}(\bm{a}\vert \bm{a}')A^{(\bm{s})}_{Q}(\bm{a}'\vert \bm{b}).
}
Since it follows that $m<n$ and $Q^{-1}(m) = \nu< Q^{-1}(n)=\mu$ from the assumptions,
$A^{(\bm{s})}_{Q}(\bm{a}' \vert \bm{b})$ does not contain the factor $Y(s_m-s_n)$.
Therefore one has
\eq{
\Res_{s_n = s_m+2iu} A^{(\bm{s})}_P(\bm{a}\vert \bm{b}) = \sum_{\bm{a}' \in \{\uparrow,\downarrow\}^N } 
\Res_{s_n = s_m +2iu}[ A^{(\bm{s}')}_{\Pi_{1,\mu}}(\bm{a}\vert \bm{a}')] \times A^{(\bm{s})}_{Q}(\bm{a}'\vert \bm{b})|_{s_n = s_m+2iu}.
}
Furthermore, since $l < m$ and 
$Q^{-1}(l)=1 < Q^{-1}(m)=\nu$, 
the factor $A^{(\bm{s})}_{Q}(\bm{a}'\vert \bm{b})\big|_{s_n = s_m+2iu}$
does not become singular at $s_m = s_l + 2iu$.
Thus it suffices to show that
$\Res_{s_n = s_m+2iu} A_{\Pi_{1,\mu}}(\bm{a}\vert \bm{b})$
contains the factor~\eqref{eq:triangle_1},
while the remaining factors are regular at $s_m = s_l +2iu$.
This can be proved by the direct calculation as follows,
\eqnn{
\Res_{s_n = s_m +2iu}A^{(\bm{s}')}_{\Pi_{1,\mu}}(\bm{a}\vert \bm{a}') 
&=
\bra{\bm{a}} \prod_{1\leq j \leq \nu-2}^{\longrightarrow} Y_{j,j+1}(s_{P(j+1)} - s_m -2iu)
\prod^{\longleftarrow}_{\nu+1 \leq j \leq \mu-1} Y_{j,j+1}(s_{l} - s_{P(j)})
\\
&\times 
Y_{\nu,\nu+1}(s_l - s_m ) Y_{\nu-1,\nu}(s_l - s_m-2iu) [-2iu(1-\Pi_{\nu,\nu+1})]
\\
&\times 
\prod_{1\leq j \leq \nu-2}^{\longleftarrow}Y_{j,j+1}(s_{l} - s_{P(j+1)})
\prod_{\nu+1 \leq j \leq \mu-1}^{\longrightarrow} Y_{j,j+1}(s_{P(j)} - s_m -2iu)\ket{\bm{a}'},
}
where $\overrightarrow{\prod}$ denotes the ordered product with increasing site indices from left to right, while $\overleftarrow{\prod}$
denotes the product with decreasing indices from left to right.

\subsection{Derivation of \eqref{eq:(n,m)(n,l)}}
We prove the following relation for $1\leq l < m <n$ and $P\in S_N^{(n)}$,
\eq{
\Res_{s_m = s_l }\big[\Res_{s_n = s_m+2iu}[A^{(\bm{s})}_P (\bm{a} \vert \bm{b})]\big] =
\Res_{s_m = s_l }\big[\Res_{s_n = s_m+2iu}[A^{(\bm{s})}_{\Pi_{l,m}P} (\bm{a} \vert \bm{b})]\big].
}

Let us set $\mu =P^{-1}(l)$ and $\nu =P^{-1}(m)$ for notational convenience and assume that $\mu < \nu$ without loss of generality.
Then, by noting $\Pi_{l,m}P = P\Pi_{\mu,\nu}$ and using \eqref{eq:A_P_product}, 
we can prove the above relation as follows,
\eqsnn{
&\Res_{s_m = s_l }\big[\Res_{s_n = s_m+2iu}[A^{(\bm{s})}_{\Pi_{l,m}P} (\bm{a} \vert \bm{b})]\big] =
\sum_{\bm{a}'\in \{\downarrow,\uparrow\}^N}
\Res_{s_m = s_l }\big[\Res_{s_n = s_m+2iu}[A^{(\bm{s})}_{P} (\bm{a}' \vert \bm{b})]\big]
\\
&\times
\bra{\bm{a}}
\Big(\prod_{\mu\leq j\leq \nu-2}^{\longrightarrow}Y_{j, j+1}(s_{P(j+1)}-s_{m} ) \Big)Y_{\nu-1,\nu}(s_l-s_m)
\Big(\prod_{\mu\leq  j \leq \nu-2}^{\longleftarrow} Y_{j,j+1}(s_l-s_{P(j+1)}) \Big)\ket{\bm{a}'}|_{s_m= s_l}
\\
&= \sum_{\bm{a}'\in\{\downarrow,\uparrow\}^N} 
\braket{\bm{a}\vert\bm{a}'}
\Res_{s_m = s_l }\big[\Res_{s_n = s_m+2iu}[A^{(\bm{s})}_{P} (\bm{a}' \vert \bm{b})]\big]
\\
& = \Res_{s_m = s_l }\big[\Res_{s_n = s_m+2iu}[A^{(\bm{s})}_{P} (\bm{a} \vert \bm{b})]\big],
}
where we used $Y_{j,j+1}(0)= I$ and $Y_{j,j+1}(\lambda)Y_{j,j+1}(-\lambda)=I$ in the second line.

\bibliographystyle{unsrturl}
\bibliography{ref}

\end{document}